\documentclass[11pt]{article}

\usepackage[margin = 1in]{geometry}
\usepackage{graphicx}              
\usepackage{amsmath}               
\usepackage{amsfonts}              
\usepackage{amsthm}                
\usepackage{amssymb}
\usepackage{mathrsfs}
\usepackage{url}
\usepackage{xcolor}
\usepackage{hyperref}
\usepackage{algpseudocode}
\usepackage[plain]{algorithm}
\usepackage{enumitem}
\usepackage{authblk}
\usepackage{tcolorbox}
\usepackage{mathtools}
\usepackage{bm}
\usepackage{dsfont}
\usepackage{tikz}
\usetikzlibrary{quantikz}

\hypersetup{
	unicode = true,
	colorlinks = true,
	citecolor = blue,
	filecolor = blue,
	linkcolor = blue,
	urlcolor = blue,
	pdfstartview = {FitH},
}

\theoremstyle{plain}
\newtheorem{theorem}{Theorem}
\newtheorem{lemma}[theorem]{Lemma}
\newtheorem{corollary}[theorem]{Corollary}

\theoremstyle{definition}
\newtheorem{definition}[theorem]{Definition}

\newtheorem*{remark}{Remark}

\algrenewcommand{\Require}{\item[\textbf{Input:}]}
\algrenewcommand{\Ensure}{\item[\textbf{Output:}]}

\newcommand{\tildO}{\tilde{O}}

\newcommand{\Romnum}[1]{\uppercase\expandafter{\romannumeral #1}}

\DeclareMathOperator{\tr}{Tr} 
\DeclareMathOperator{\negl}{negl} 

\DeclareMathOperator{\poly}{poly}

\DeclareMathOperator{\entpy}{H}
\let\hom\relax
\DeclareMathOperator{\hom}{Hom}
\DeclareMathOperator{\qft}{F}
\DeclareMathOperator{\E}{\mathbb{E}}

\DeclarePairedDelimiter{\abs}{\lvert}{\rvert}
\let\ket\relax
\DeclarePairedDelimiter{\ket}{\lvert}{\rangle}
\let\bra\relax
\DeclarePairedDelimiter{\bra}{\langle}{\rvert}
\DeclarePairedDelimiter{\lrang}{\langle}{\rangle}
\DeclarePairedDelimiter{\opnorm}{\lVert}{\rVert}

\def\C{\mathbb{C}}

\def\N{\mathbb{N}}
\def\R{\mathbb{R}}
\def\Z{\mathbb{Z}}

\def\lwe{\mathsf{LWE}}
\def\edcp{\mathsf{EDCP}}
\def\gen{\mathsf{Gen}}
\def\enc{\mathsf{Enc}}
\def\dec{\mathsf{Dec}}
\def\X{\mathcal{X}}
\def\SX{\mathcal{S(X)}}
\def\U{\mathcal{U}}
\def\PosSemi{\mathrm{Pos}}

\title{Efficient Quantum Public-Key Encryption \\ From Learning With Errors}

\author{
	Javad Doliskani\thanks{Department of Computer Science, Ryerson University,
	(\tt{javad.doliskani@ryerson.ca}).}
}

\date{}
\allowdisplaybreaks

\begin{document}
\maketitle

\begin{abstract}
    Our main result is a quantum public-key encryption scheme based on the Extrapolated Dihedral Coset problem (EDCP) which is equivalent, under quantum polynomial-time reductions, to the Learning With Errors (LWE) problem. For limited number of public keys (roughly linear in the security parameter), the proposed scheme is information-theoretically secure. For polynomial number of public keys, breaking the scheme is as hard as solving the LWE problem. The public keys in our scheme are quantum states of size $\tildO(n)$ qubits. The key generation and decryption algorithms require $\tildO(n)$ qubit operations while the encryption algorithm takes $O(1)$ qubit operations.
\end{abstract}

\newpage

\section{Introduction}
\label{sec:intro}

The application of quantum information in cryptography goes back to the work of Wiesner \cite{wiesner1983conjugate} who proposed the first quantum cryptographic tool called \textit{Conjugate Coding}. Remarkably, the idea of conjugate coding is still used in different forms in many modern protocols in quantum cryptography. Quantum cryptography, however, gained much attraction after the introduction of Quantum Key Distribution (QKD) by Bennett and Brassard \cite{bennett1983quantum, bennett1984quantum}. It was later proved by Lo and Chau \cite{lo1999unconditional} and Mayers \cite{mayers2001unconditional} that QKD is information-theoretically secure. A more accessible proof of security, based on error correcting codes, was  given by Shor and Preskill \cite{shor2000simple}.

Although in theory QKD provides perfect security, its real-world implementations are not (and perhaps will not be) ideal. That means the QKD implementations, like other cryptographic implementations, are vulnerable to side-channel attacks, e.g., see \cite{lydersen2010hacking}. Even if we assume that QKD provides perfect security in practice, there are many other important cryptographic tasks, such as bit commitment, multiparty computation and oblivious transfer, that are not addressed by key distribution. In fact, it was proved by Mayers \cite{mayers1997unconditionally} and Lo and Chau \cite{lo1997quantum} that unconditionally secure quantum bit commitment is impossible. The impossibility of information-theoretically secure two-party computation using quantum communication was also later proved by Colbeck \cite{colbeck2007impossibility}. Such schemes can be made secure if the adversary is assumed to have bounded computational power or limited storage. Computational assumption are therefore still needed and are important in quantum cryptography. In particular, the necessity of computational assumptions in quantum \textit{public-key cryptography}, which is an increasingly important area in quantum cryptography, needs to be further investigated.

The principles of quantum public-key cryptography are close adaptations of those of classical public-key cryptography. In a quantum public-key scheme, a pair of keys $(sk_A, pk_A)$ is associated to each user $A$, the secret key $sk_A$ which is only known to $A$, and the public key $pk_A$ which is published by $A$ and it can be accessed by everyone. The pair of keys are generated by an efficient key generation algorithm. Like classical public-key schemes, quantum public-key schemes are modelled based on trapdoor one-way functions. Informally, a one-way function is a function that is easy to compute but hard to invert. A trapdoor one-way function is a one-way function $f$ to which some information $k$, called the trapdoor, can be associated in a way that anyone with the knowledge of $k$ can easily invert $f$ \cite{boneh2015graduate}. In the quantum setting, $f$ is a mapping $\ket{\alpha} \mapsto \ket{f_\alpha}$ from the space of secret keys to the space of public keys. The secret key $\ket{\alpha}$ can be a classical or quantum state, and the public key $\ket{f_\alpha}$ is a quantum state. 

The three main constructions in quantum public-key cryptography are public-key encryption, digital signature and public-key money. In this work, we focus on quantum public-key encryption. We refer the reader to \cite{gottesman2001quantum} for quantum digital signatures, and \cite{aaronson2009quantum, aaronson2012quantum, farhi2012quantum} for quantum money. In a public-key encryption scheme, user $B$ can send a secret message $m$ to $A$ by encoding $m$ into a ciphertext $c$ using $A$'s public key $pk_A$ and a public encryption algorithm. Upon receiving the ciphertext $c$, user $A$ decrypts $c$ using her private key $sk_A$ and a public decryption algorithm.

\subsection{Previous constructions}

Kawachi \textit{et al.}~\cite{kawachi2005computational, kawachi2012computational} proposed a quantum public-key encryption scheme based on the Graph Isomorphism (GI) problem. For a security parameter $n$, they consider quantum states over the symmetric group $S_n$. The (secret key, public key) pair in the proposed scheme is $(sk, pk) = (s, \ket{\psi_s})$ where $s \in S_n$ is such that $s^2 = 1$, and $\ket{\psi_s} = (\ket{t} + \ket{ts}) / \sqrt{2}$ for a random $t \in S_n$. A single bit $b$ of classical information is then encrypted as $(\ket{t} + (-1)^b\ket{ts}) / \sqrt{2}$. They prove their scheme is secure assuming the hardness of the GI problem. However, GI is not considered as a standard quantum-secure assumption, especially after the recent breakthrough by L{\'a}szl{\'o} Babai \cite{babai2016graph} which shows that GI can be solved classically in quasipolynomial time.

An encryption scheme based on rotations of $1$-qubit states was proposed by Nikolopoulos \cite{nikolopoulos2008applications}. Their scheme is based on the trapdoor function $s \mapsto \ket{\psi_s(\theta_n)}$ where $0 \le s < 2^n$ is an integer, $\theta_n = \pi / 2^{n - 1}$ and $\ket{\psi_s(\theta_n)} = \cos(s\theta_n /2)\ket{0} + \sin(s\theta_n / 2)\ket{1}$. The encryption algorithm, however, was deterministic and could be broken by a simple attack proposed in \cite{nikolopoulos2009deterministic}. A randomized version of the encryption algorithm was proposed, in the same paper, to minimize the success probability of the aforementioned attack. To encrypt a classical message, the randomized algorithm uses significantly larger public-key states than the original algorithm.

\subsection{Overview and results}

\paragraph{Extrapolated dihedral cosets.}
Let $q$ be a positive integer. A dihedral coset over the additive group $\Z_q$ is a quantum state defined as $\ket{\psi_x} := (\ket{0}\ket{x} + \ket{1}\ket{x + s}) / \sqrt{2}$ where $x \in \Z_q$ is random and $s \in \Z_q$ is random and fixed. The Dihedral Coset Problem (DCP) is the problem of recovering (the secret) $s$, given dihedral cosets states $\ket{\psi_x}$ for different random values of $x$. DCP arises naturally when one uses the ``standard method'' to solve the hidden shift problem over the group $\Z_q$ \cite{childs2010quantum}. More precisely, the states $\ket{\psi_x}$ are obtained by performing a measurement on a superposition over the dihedral group $D_q = \Z_q \rtimes \Z_2$. Regev \cite{regev2004quantum} gave the first connection between DCP and standard lattice problems. Informally, if the number of states in DCP is bounded by $\poly(\log q)$, then DCP is (quantumly) at least as hard as $\poly(\log q)$-unique-SVP. 

In a recent work, Brakerski \textit{et al.}~\cite{brakerski2018learning} proved that the Learning With Errors ($\lwe$) problem is equivalent, under quantum polynomial-time reductions, to an extended version of DCP, which they call the Extrapolated Dihedral Coset Problem ($\edcp$). The $\lwe$ problem was proposed by Regev \cite{regev2005lattices, regev2009lattices}, and it has been used as the underlying security assumption in a large number of cryptographic schemes. For parameters $n, q \in \Z$ and $\alpha \in (0, 1)$, $\lwe_{n, q, \alpha}$ is the problem of recovering $\bm{s} \in \Z_q^n$ from samples of the form $(\bm{a}, \lrang{\bm{a}, \bm{s}} + e)$ where $\bm{a} \in \Z_q^n$ is uniformly random and $e$ is sampled from the discrete Gaussian distribution $\mathcal{D}_{\Z, \alpha q}$ of standard deviation $\alpha q$. The (uniform) $\edcp$ is a generalization of DCP where coset states are over the group $\Z_q^n$ and the number of terms in the coset states is parameterized by an integer $2 \le r \le q$. An $\edcp$ state is written as
\begin{equation}
    \label{equ:edcp-state-intro}
    \frac{1}{\sqrt{r}} \sum_{j = 0}^{r - 1}\ket{j}\ket{\bm{x} + j\bm{s}},
\end{equation}
where $\bm{s}$ is fixed and $\bm{x} \in \Z_q^n$ is uniformly random. The $\edcp_{n, q, r}$ is then to recover $\bm{s}$, given states of the form \eqref{equ:edcp-state-intro}. In the $\edcp$ and $\lwe$ equivalence given in \cite{brakerski2018learning}, the parameter $r$ is proportional to $1 / \alpha$. Intuitively, one would expect that $\edcp$ becomes easier for larger $r$, since for smaller $\alpha$ the $\lwe$ problem becomes easier. In fact, we show that $\edcp_{n, q, r}$ for any $r$ can be efficiently reduced to $\edcp_{n, q, 2}$ (Lemma \ref{lem:self-rd}). On the other hand, in the extreme case $r = q$, $\edcp$ can be solved in polynomial time (Remark \ref{rmk:edcp-extreme}). 

\paragraph{A search-to-decision reduction for $\edcp$.}
The decision-$\edcp$, as was defined in \cite{brakerski2018learning}, is the problem of distinguishing between states of the form \eqref{equ:edcp-state-intro} and the same number of states of the form $\ket{j}\ket{\bm{x}}$ for uniformly random $(j, \bm{x}) \in \Z_r \times \Z_q^n$. It was pointed out in \cite{brakerski2018learning} that there is a polynomial time search-to-decision reduction for $\edcp$ via $\lwe$. More precisely, search-$\edcp$ can be reduced to decision-$\edcp$ via the following sequence of polynomial-time reductions
\begin{equation}
    \label{equ:redu-seq}
    \text{search-}\edcp \le \text{search-}\lwe \le \text{decision-}\lwe \le \text{decision-}\edcp.
\end{equation}
In Section \ref{sec:old-search-dec}, we give a direct reduction from search-$\edcp$ to decision-$\edcp$ (Theorem \ref{thm:old-search-decision}) which works for a large class of moduli $q$.

\paragraph{A new decision-$\edcp$.}
The original decision-$\edcp$, explained above, is obtained from decision-$\lwe$ by following the same procedure used to establish the equivalence between search-$\lwe$ and search-$\edcp$. For such a decision problem for $\edcp$, however, there is no known general reduction from the search-$\edcp$; The search-to-decision reductions are either obtained using the sequence \eqref{equ:redu-seq} or using Theorem \ref{thm:old-search-decision}. In the former case, one has to rely on the search-to-decision reductions for $\lwe$, e.g., \cite{regev2009lattices, brakerski2013classical, micciancio2012trapdoors}, which are not general enough in the sense that they either incur non-negligible loss in $\lwe$ parameters, or work only for special forms of the modulus $q$. In the latter case, the parameter $r$ has to satisfy some constraints with respect to $q$.

One could argue that the equivalence between the search and decision problems for $\edcp$ is of less importance due to the recent result of Peikert \textit{et al.} \cite{peikert2017pseudorandomness} who proved that there is a polynomial-time quantum reduction from standard lattice problems directly to decision-$\lwe$. Their proof works for any modulus $q$. In any case, another issue is that it is not clear to us how to base \textit{efficient} cryptographic primitives on the original decision-$\edcp$, see the discussion in Section \ref{sec:public-key-enc}. To resolve these issues, we propose a new decision problem for $\edcp$. Informally, the new decision problem asks to distinguish between states of the form \eqref{equ:edcp-state-intro} and states of the form 
\[ \frac{1}{\sqrt{r}} \sum_{j = 0}^{r - 1} \omega_p^{jt} \ket{j}\ket{\bm{x} + j\bm{s}}, \]
where $p$ is a prime divisor of $q$ and $t \in \Z_p {\setminus} \{0\}$ is uniformly random. This new decision problem allows us to accomplish the following:
\begin{itemize}
\item We prove that for any modulus $q$ with $\poly(n)$-bounded prime factors, there is a quantum polynomial-time reduction from solving search-$\edcp$ to solving decision-$\edcp$.
\item We build an efficient quantum public-key encryption scheme based on decision-$\edcp$.
\end{itemize}
 
\paragraph{Quantum public-key cryptosystem.}
In Section \ref{sec:public-key-enc}, we build a quantum public-key encryption scheme from $\edcp$. The idea behind the encryption is very simple. The public key is a single $\edcp$ state as in \eqref{equ:edcp-state-intro}. To encrypt a classical bit $b \in \{0,1\}$, a unitary transform is applied to the public key to encode the value $bt$ into the phase, where $t \in \Z_p$ is uniformly random and nonzero. The ciphertext is the state
\[ \frac{1}{\sqrt{r}} \sum_{j = 0}^{r - 1} \omega_p^{jbt} \ket{j}\ket{\bm{x} + j\bm{s}}. \]
Here, $r$ and $p$ are divisors of $q$, and $r = p^{s'}$ for some integer $s' \ge 1$. The decryption algorithm is essentially a scalar multiplication over $\Z_q^n$ and an application of the quantum Fourier transform $\qft_r$ over $\Z_r$. This scheme is very efficient in terms of both computation and public key size. An even more efficient instantiation of this scheme is obtained by setting $p = 2$, $r = 2^{s'}$, and $q = 2^s$ such that $q = \poly(n)$ and $s' \ll s$. In this case, the size of the public key is $\tildO(n)$ qubits. The encryption algorithm takes $O(1)$ qubit operations, and the key generation and decryption algorithms take $\tildO(n)$ qubit operations.

\paragraph{Security beyond Holevo's bound.} 
A fundamental difference between classical and quantum public-key cryptosystems is that in classical systems the public key can be copied arbitrarily many times, while in quantum systems making even two copies of the same public key is generally impossible. This is a consequence of the \textit{no-cloning} theorem. Also, in the quantum setting, the amount of information one can extract from a public key is bounded by a certain quantity which depends on the parameters of the system. This is a consequence of Holevo's theorem (Theorem \ref{thm:holevo}) which says that the accessible information of an ensemble is bounded by the $\chi$ quantity of the ensemble. Therefore, according to Holevo's theorem, the secret key cannot be recovered as long as the number of copies of the public key stays below Holevo's bound.

A feature common to all previous quantum public-key systems is that they rely on Holevo's bound for information-theoretic security, but do not provide a security proof based on a standard assumption beyond Holevo's bound. This puts a sever limitation on the total number of public keys published at any time. In Section \ref{sec:hardness}, we prove that our cryptosystem is information-theoretically secure when the number of public keys is bounded by $O(n\log q / \log r)$. Beyond that, when the number of public keys is $\poly(n)$, breaking the scheme is as hard as solving $\lwe$.


\section{Preliminaries}
\label{sec:preli}

\subsection{Quantum Computation}

Our notations for quantum information mostly follow those of \cite{watrous2018theory}. The classical state of a register $\mathsf{X}$ is represented by a finite alphabet, say $\Sigma$. If the registers $\mathsf{X}_1, \dots, \mathsf{X}_n$ are represented by alphabets $\Sigma_1, \dots, \Sigma_n$ then the classical state of the tuple $(\mathsf{X}_1, \dots, \mathsf{X}_n)$ is represented by $\Sigma_1 \times \cdots \times \Sigma_n$. The complex Euclidean space associated with the register $\mathsf{X}$ is denoted by $\C^\Sigma$. For a complex Euclidean space $\X$, denote the unit sphere in $\X$ by $\SX$. A linear operator $\rho$ acting on $\X$ is called a density operator if $\rho$ is positive semidefinite with trace equal to $1$. The quantum state of the register $\mathsf{X}$ is represented by the set of density operators $\mathrm{D}(\X)$.

We will use the Dirac notation for the elements of $\SX$. In particular, we denote the column vector $x \in \SX$ by $\ket{x}$ and the row vector $x^*$ by $\bra{x}$. A state $\rho$ is called pure if it can be written as $\rho = \ket{x}\bra{x}$, in which case we will simply write the state as $\ket{x}$. By the spectral theorem, every state $\rho$ is a linear combination of pure states. Therefore, a quantum state can also be represented by a linear combination
\[ \sum_{x} \alpha_x \ket{x}, \hspace*{1mm} \sum_{x} \abs{\alpha_x}^2 = 1. \]
We shall alternate between these equivalent representations of quantum states throughout this paper. The density operator representation is particularly useful when the underlying quantum state is not completely known. For example, if we only know that the system is in the state $\ket{\psi_x}$ with probability $p_x$ then the state of the system is described by the density operator
\[ \rho = \sum_{x} p_x \ket{\psi_x}\bra{\psi_x} = \E_x \Big[ \ket{\psi_x} \bra{\psi_x} \Big]. \]
More, generally, the density operator corresponding to a probability distribution $\gamma: \SX \rightarrow [0, 1]$ is defined as
\[ \rho_\gamma = \int_{\ket{\phi} \in \SX} \ket{\phi}\bra{\phi} d\gamma(\ket{\phi}) = \E_{\ket{\phi} \in \gamma} \Big[ \ket{\phi}\bra{\phi} \Big]. \]
For quantum public-key cryptography we will need a formal notion of quantum state discrimination. In particular, we need to formally define the notion of computational (in)distinguishability of quantum states. For our purposes, it is more convenient to define computational distinguishability for probability distributions over quantum states. The following is adapted from \cite[3.3]{watrous2009zero}.

\begin{definition}
    Let $\X$ be a complex Euclidean space, and let $\gamma, \mu: \SX \rightarrow [0, 1]$ be probability distributions. Then $\gamma$ is said to be $(s, \epsilon)$-distinguishable from $\mu$ if there is a quantum measurement circuit $Q$ of size $s$ such that
    \[ \abs[\Big]{ \Pr_{\rho \in \gamma}[Q(\rho) = 1] - \Pr_{\rho \in \mu}[Q(\rho) = 1]} \ge \epsilon. \]
\end{definition} 

Two distributions $\gamma, \mu$ are $(s, \epsilon)$-indistinguishable if they are not $(s, \epsilon)$-distinguishable.

\begin{definition}
    For each $n \in \N$, let $\X_n$ be a complex Euclidean space and let $\gamma_n, \mu_n: \mathcal{S}(\X_n) \rightarrow [0, 1]$ be probability distributions. Then the two ensembles $\{ \gamma_n \}_{n \in \N}$ and $\{ \mu_n \}_{n \in \N}$ are said to be polynomially quantum indistinguishable if for all polynomially bounded functions $s, p: \N \rightarrow \N$, the distributions $\gamma_n$ and $\mu_n$ are $(s(n), 1 / p(n))$-indistinguishable for almost all $n \in \N$.
\end{definition}

Two ensembles are called quantum computationally indistinguishable if they are polynomially quantum indistinguishable. The advantage of a polynomial-time quantum algorithm $Q$ in distinguishing between the distributions $\gamma_n$ and $\mu_n$ is defined as
\[ \delta_Q(\gamma_n, \mu_n) = \abs[\Big]{ \Pr_{\rho \in \gamma_n}[Q(\rho) = 1] - \Pr_{\rho \in \mu_n}[Q(\rho) = 1] }. \]
Two ensembles $\{ \gamma_n \}$ and $\{ \mu_n \}$ are then quantum computationally indistinguishable if $\delta_Q(\gamma_n, \mu_n) = \negl(n)$ for all such $Q$ and almost all $n$.

\subsection{Error reduction}
\label{sec:err-red}

We can abstractly define the advantage of an algorithm $A$, regardless of $A$ being quantum or classical, in distinguishing between two probability distribution $P_1$ and $P_2$ as
\[ \delta_A(P_1, P_2) = \abs[\Big]{ \Pr_{x \in P_1}[A(x) = 1] - \Pr_{x \in P_2}[A(x) = 1] }. \]
Two ensembles of distributions $\{ P_{1, n} \}$ and $\{ P_{2, n} \}$ are said to be polynomial-time indistinguishable if for any polynomial-time algorithm $A$ and any $\poly(n)$-bounded function $p$ we have $\delta_A(P_{1, n}, P_{2, n}) \le 1 / p(n)$ for large enough $n$. The following lemma follows from the triangle inequality.

\begin{lemma}[Hybrid lemma]
    \label{lem:hybrid}
    Let $P_1, \dots, P_k$ be a sequence of probability distributions. Assume that $\delta_A(P_1, P_k) \ge \epsilon$ for some polynomial-time algorithm $A$. Then $\delta_A(P_i, P_{i + 1}) \ge \epsilon / k$ for some $1 \le i < k$.
\end{lemma}

Suppose an algorithm $A$ can distinguish between two distributions $P_1$ and $P_2$ with non-negligible advantage. A common technique to amplify the distinguishing advantage of $A$ is to sample enough times from the input distribution and then decide based on majority. A brief description of this technique, which we shall use several times in this paper, is as follows. First, we need the following well-known tail inequality.

\begin{lemma}[Hoeffding \cite{hoeffding1963probability}]
    \label{lem:hoeffding}
    Let $X_1, \dots, X_n$ be independent random variables with $X_i \in [a_i, b_i]$, and let $S = X_1, \cdots + X_n$. Then
    \begin{align*}
        \Pr[S - \E[S] \ge t] & \le e^{-2t^2 / \sum_i^n (b_i - a_i)^2}, \text{ and} \\
        \Pr[S - \E[S] \le -t] & \le e^{-2t^2 / \sum_i^n (b_i - a_i)^2}.
    \end{align*}
\end{lemma}

Now, assume $\delta_A(P_1, P_2) \ge 1 / p(n)$ for some polynomial $p(n)$, and let $P$ be the input distribution. Draw $m = 2np(n)^2$ samples from $P$, and let $X_i$ be a random variable representing the output of $A$ on input the $i$-th sample. Here, $X_i = 0$ (resp., $X_i = 1$) means $A$ has recognized the $i$-th sample to be from $P_1$ (resp., $P_2$). Let $S = X_1 + \cdots + X_m$. If $P = P_2$ then from the bound on $\delta_A$ we have $\E[S] \ge m / 2 + np(n)$. By Hoeffdings's inequality,
\begin{align*}
    \Pr\Big[ S \le \frac{1}{2} (m + np(n)) \Big]
    & = \Pr\Big[ S - \frac{m}{2} - np(n) \le -\frac{1}{2}np(n) \Big] \\
    & \le \Pr\Big[ S - \E[S] \le -\frac{1}{2}np(n) \Big] \\
    & \le e^{-n / 4}.
\end{align*}
Similarly, if $P = P_1$ then
\begin{align*}
    \Pr\Big[ S \ge \frac{1}{2} (m - np(n)) \Big]
    & \le \Pr\Big[ S - \E[S] \ge \frac{1}{2}np(n) \Big] \\
    & \le e^{-n / 4}.
\end{align*}
Therefore, by running $A$ on $m$ samples and counting the number of $1$'s we can tell, with probability exponentially close to $1$, whether $P = P_1$ or $P = P_2$.

\subsection{Learning With Errors}

In the following, we briefly review the Learning With Errors ($\lwe$) problem and the Extrapolated Dihedral Coset Problem ($\edcp$). In \cite{brakerski2018learning}, $\edcp$ refers to a general class of problems from which two specific instances are studied in detail: the $U$-$\edcp$ and $G$-$\edcp$ which are the \textit{Uniform} and \textit{Gaussian} $\edcp$, respectively. In this work, we will only study uniform-$\edcp$, and simply refer to it as $\edcp$.

Let $n \ge 1$ and $q = q(n) \ge 2$ be integers, and let $\chi$ be a probability distribution over $\Z$. For a random fixed $\bm{s} \in \Z_q^n$, denote by $A_{\bm{s}, \chi}$ the probability distribution over $\Z_q^n \times \Z_q$ defined as follows: choose $\bm{a} \in \Z_q^n$ uniformly at random, choose $e$ according to $\chi$ and output $(\bm{a}, \lrang{\bm{a}, \bm{s}} + e)$.

\begin{definition}[LWE, Search]
The search-$\lwe_{n, q, \chi}$ is the problem of recovering $\bm{s}$ given samples from the distribution $A_{\bm{s}, \chi}$. An algorithm $Q$ is said to solve $\lwe_{n, q, \chi}$ if $Q$  outputs $\bm{s}$ with probability at least $1 / \poly(n\log q)$ and has running time at most $\poly(n \log q)$.
\end{definition}

\begin{definition}[LWE, Decision]
    The decision-$\lwe_{n, q, \chi}$ problem is to distinguish between the distribution $A_{\bm{s}, \chi}$ and the uniform distribution over $\Z_q^n \times \Z_q$. An algorithm $Q$ is said to solve the desicion-$\lwe_{n, q, \chi}$ if it succeeds with advantage at least $1 / \poly(n\log q)$ and has running time at most $\poly(n\log q)$. 
\end{definition}

The distribution $\chi$ is called the error distribution and is usually chosen to be $\mathcal{D}_{\Z, \alpha q}$, the discrete Gaussian distribution centered around zero with standard deviation $\alpha q$. The parameter $\alpha \in (0, 1)$ is called the error rate. We sometimes write $\lwe_{n, q, \alpha}$ instead of $\lwe_{n, q, \chi}$ for simplicity. Let $n \ge 1$ and $q \ge 2$ be as above and let $r = r(n) < q$ be a positive integer. Let $\Sigma = \Z_r \times \Z_q^n$ and define the complex Euclidean space $\X = \C^\Sigma$. For a fixed uniformly random $\bm{s} \in \Z_q^n$ define the probability distribution $\mu_{\bm{s}, r}: \SX \rightarrow [0, 1]$ as follows: choose $\bm{x} \in \Z_q^n$ uniformly at random and output the state
\[ \ket{\phi_{\bm{s}, r}(\bm{x})} = \frac{1}{\sqrt{r}} \sum_{j = 0}^{r - 1}\ket{j}\ket{\bm{x} + j\bm{s}}. \]
If we only have access to the output of $\mu_{\bm{s}, r}$, e.g., if $\bm{x}$ is unknown, then the quantum system corresponding to the state $\ket{\phi_{\bm{s}, r}(\bm{x})}$ is described by the density operator
\[ \rho_{\bm{s}, r} = \frac{1}{q^n} \sum_{\bm{x} \in \Z_q^n} \ket{\phi_{\bm{s}, r}(\bm{x})} \bra{\phi_{\bm{s}, r}(\bm{x})} = \E_{\bm{x} \in \U(\Z_q^n)} \Big[ \ket{\phi_{\bm{s}, r}(\bm{x})} \bra{\phi_{\bm{s}, r}(\bm{x})} \Big]. \]
Therefore, the output of the distribution $\mu_{\bm{s}, r}$ is always described by the same state $\rho_{\bm{s}, r}$. In other words, a sample from the distribution $\mu_{\bm{s}, r}$ is a \textit{copy} of the state $\rho_{\bm{s}, r}$.

\begin{definition}[EDCP, Search]
    Let $n$, $q$ and $r$ be defined as above. The search-$\edcp_{n, q, r}$ is the problem of recovering $\bm{s}$ given samples from the distribution $\mu_{\bm{s}, r}$. A quantum algorithm $Q$ is said to solve $\edcp_{n, q, r}$ if it outputs $\bm{s}$ with probability at least $1 / \poly(n\log q)$ and has running time at most $\poly(n\log q)$.
\end{definition}

\begin{definition}[EDCP, Decision]
    \label{def:d-edcp}
    Let $n$, $q$ and $r$ be defined as above. Define the probability distribution $\gamma_r: \SX \rightarrow [0, 1]$ by choosing $(j, \bm{x}) \in \Z_r \times \Z_q^n$ uniformly at random and outputting the state $\ket{j}\ket{\bm{x}}$. The decision-$\edcp_{n, q, r}$ is the problem of distinguishing between the distributions $\mu_{\bm{s}, r}$ and $\gamma_r$. A quantum algorithm $Q$ is said to solve the decision-$\edcp_{n, q, r}$ if it succeeds with advantage at least $1 / \poly(n\log q)$ and has running time at most $\poly(n\log q)$.
\end{definition}

The density operator corresponding to the output of the distribution $\gamma_r$ in Definition \ref{def:d-edcp} is
\[ \rho = \frac{1}{rq^n} \sum_{j = 0}^{r - 1} \sum_{\bm{x} \in \Z_q^n}  \ket{j}\ket{\bm{x}} \bra{j}\bra{\bm{x}} = \E_{(j, \bm{x}) \in \U(\Z_r \times \Z_q^n)} \Big[ \ket{j}\ket{\bm{x}} \bra{j}\bra{\bm{x}} \Big] = \frac{1}{rq^n} \mathds{1}_{\X}, \]
which is the maximally mixed state over the space $\X$. Therefore, decision-$\edcp_{n, q, r}$ is the problem of distinguishing between the same number of copies of the states $\rho_{\bm{s}, r}$ and $\mathds{1}_{\X} / (rq^n)$. For an integer $m > 0$, we denote by $\edcp_{n, q, r}^m$ the $\edcp$ problem in which the number of samples from $\mu_{\bm{s}, r}$ is bounded by $m$. The following theorem establishes a polynomial-time equivalence between $\lwe$ and $\edcp$.

\begin{theorem}[\cite{brakerski2018learning}]
    \label{thm:lwe-edcp}
    Let $\chi$ be a discrete Gaussian distribution centered around zero with standard deviation $\alpha q$. There is a polynomial-time quantum reduction from $\lwe_{n, q, \chi}$ to $\edcp_{n, q, r}^m$ with $m = \poly(n\log q)$ and $r = \poly(n\log q) / \alpha$. Conversely, for the same parameter relationship up to $\poly(n\log q)$ factors, there is a polynomial-time quantum reduction from $\edcp$ to $\lwe$. 
\end{theorem}

It is important to note that the equivalence in Theorem \ref{thm:lwe-edcp} holds only when the number of $\edcp$ samples is polynomially bounded, i.e., $m = \poly(n\log q)$. Giving such an equivalence for arbitrary $m$ is an open problem, see the discussion in Section \ref{sec:hardness-poly}.


\section{A Search to Decision Reduction}
\label{sec:old-search-dec}

In this section, we give a search-to-decision reduction for $\edcp$. The reduction works for a large class of moduli $q$. The technique we use is inspired by the one in \cite{micciancio2012trapdoors} for a search-to-decision reduction for $\lwe$. We need the following lemma, which shows the self-reducibility of $\edcp$.

\begin{lemma}
    \label{lem:self-rd}
    For any $r' \le r$, given access to the distribution $\mu_{\bm{s}, r}$, we can efficiently sample from the distribution $\mu_{\bm{s}, r'}$. In particular, there is an efficient reduction from $\edcp_{n, q, r}$ to $\edcp_{n, q, r'}$.
\end{lemma}
\begin{proof}
    To keep the reduction efficient, we treat the two cases $r' > r / 2$ and $r' \le r / 2$ separately. If $r' > r / 2$ then a simple indicator function can be used to to produce samples from $\mu_{\bm{s}, r'}$. More precisely, define the function $f: [0, r) \rightarrow \{ 0, 1 \}$ by
    \[ f(x) = 
    \begin{cases}
        1 & \text{if } x < r', \\
        0 & \text{otherwise}.
    \end{cases} \]
    Then applying the transform $\ket{j}\ket{\bm{y}}\ket{0} \mapsto \ket{j}\ket{\bm{y}}\ket{f(j)}$, where $\bm{y} \in \Z_q^n$, to $\rho_{\bm{s}, r} \in \mu_{\bm{s}, r}$ and measuring the last register results in the state $\rho_{\bm{s}, r'}$ with probability at least $1 / 2$. If the measurement outcome is not $1$ then we repeat the above process.

    If $r' \le r / 2$ then we proceed as follows. Let $\X = \C^{\Z_r \times \Z_q^n}$ and define the measurement $\mu: [0, \lfloor r / r'  \rfloor] \rightarrow \PosSemi(\X)$ using the operators
    \[ \mu_a = \sum_{b = 0}^{\ell - 1} \ket{ar' + b}\bra{ar' + b} \otimes \mathds{1}, \]
    where $\ell = r'$ for $0 \le a < \lfloor r / r'  \rfloor$ and $\ell = r - r'$ for $a = \lfloor r / r'  \rfloor$. This measurement can be implemented efficiently \cite[A.8]{kaye2007introduction}. If we perform $\mu$ on a sample $\rho_{\bm{s}, r}$ from $\mu_{\bm{s}, r}$ the probability of observing the outcome $a$ is
    \begin{align*}
        \tr(\mu_a^*\mu_a \rho_{\bm{s}, r})
        & = \E_{\bm{x} \in \U(\Z_q^n)} \Big[ \tr(\mu_a \ket{\phi_{\bm{s}, r}(\bm{x})} \bra{\phi_{\bm{s}, r}(\bm{x})} \mu_a^*) \Big] \\
        & = \frac{1}{r} \sum_{b, c = 0}^{\ell - 1} \E_{\bm{x} \in \U(\Z_q^n)}\Big[ \tr(\ket{ar' + b}\ket{\bm{x} + (ar' + b)\bm{s}} \bra{ar' + c}\bra{\bm{x} + (ar' + c)\bm{s}}) \Big] \\
        & = \frac{1}{r} \sum_{b, c = 0}^{\ell - 1} \E_{\bm{x} \in \U(\Z_q^n)}\Big[ \tr(\ket{ar' + b}\bra{ar' + c} \otimes \ket{\bm{x} + (ar' + b)\bm{s}} \bra{\bm{x} + (ar' + c)\bm{s}}) \Big] \\
        & = \frac{\ell}{r},
    \end{align*}
    and the post-measurement state corresponding to this outcome is
    \begin{align*}
        \frac{\mu_a \rho_{\bm{s}, r} \mu_a^*}{(\ell / r)}
        & = \frac{1}{r} \sum_{b, c = 0}^{\ell - 1} \E_{\bm{x} \in \U(\Z_q^n)}\Big[ \ket{ar' + b}\ket{\bm{x} + (ar' + b)\bm{s}} \bra{ar' + c}\bra{\bm{x} + (ar' + c)\bm{s}} \Big] \\
        & = \frac{1}{r} \sum_{b, c = 0}^{\ell - 1} \E_{\bm{x} \in \U(\Z_q^n)}\Big[ \ket{ar' + b}\ket{\bm{x} + b\bm{s}} \bra{ar' + c}\bra{\bm{x} + c\bm{s}} \Big].
    \end{align*}
    Subtracting $ar'$ from the first register, we obtain the state $\rho_{\bm{s}, \ell}$. So if the outcome is $a \in [0, \lfloor r / r'  \rfloor)$ we obtain the state $\rho_{\bm{s}, r'}$, which is what we are looking for. Therefore, the probability of obtaining the desired state after one measurement is
    \[ \lfloor r / r' \rfloor \frac{\ell}{r} = \lfloor r / r' \rfloor \frac{r'}{r} \ge \left( \frac{r}{r'} - 1 \right)\frac{r}{r'} = 1 - \frac{r'}{r} \ge \frac{1}{2}. \]
    If the measurement outcome is $a = \lfloor r / r' \rfloor$ then we repeat the above process.
\end{proof}

It follows from the proof of Lemma \ref{lem:self-rd} that obtaining a sample from $\mu_{\bm{s}, r'}$ requires (an expected) $2$ samples from $\mu_{\bm{s}, r}$. This means $\edcp_{n, q, r}^m$ is reduced to $\edcp_{n, q, r'}^{\Theta(m)}$ regardless of the ratio between $r$ and $r'$.

\begin{theorem}
    \label{thm:old-search-decision}
    Let $q = p_1^{e_1} \cdots p_\ell^{e_\ell}$ be the prime factorization of $q$ and assume that the primes $p_i$ are of size $\poly(n)$. Let $0 < r < q$ and let $k$ be the number of primes $p_i < r$. Then there is a polynomial-time quantum reduction from solving worst-case search-$\edcp_{n, q, r}$, with overwhelming probability, to solving average-case decision-$\edcp_{n, q, r'}$, with non-negligible probability, for any $r' \le r$ such that $(r')^k \le r$, and $r' \le p_i^{e_i}$ for all $i$. 
\end{theorem}
\begin{proof}
    Let $D$ be an oracle for solving decision-$\edcp_{n, q, r'}$. The idea of the proof is to use $D$ and samples from the distribution $\mu_{\bm{s}, r}$ to recover $\bm{s} \bmod p_i^{h_i}$, with large-enough $h_i$, for each $i$, and then assemble the results using the Chinese remainder theorem to recover $\bm{s} \bmod \prod_i p_i^{h_i}$. From there, since $q / \prod_i p_i^{h_i}$ is small-enough, we can use quantum Fourier transform to recover $\bm{s} \bmod q$. We shall compute $\bm{s} \bmod p_1^{e_1}$, the algorithm is the same for the other $p_i$. Let $p = p_1$ and $e = e_1$. The proof proceeds in several steps.

    \begin{enumerate}[leftmargin = *, font = \bfseries]
    \item Sampling from $\mu_{\bm{s}, r'}$: given samples from $\mu_{\bm{s}, r}$, according to Lemma \ref{lem:self-rd}, we can efficiently sample from $\mu_{\bm{s}, r'}$. So, from now on we assume that we have access to samples from $\mu_{\bm{s}, r'}$.

    \item Building hybrid distributions: from the distribution $\mu_{\bm{s}, r'}$ we construct the distribution $\mu_{\bm{s}, r'}^k$ for all $k = 0, \dots, e$. Given a sample $\rho_{\bm{s}, r'} \in \mu_{\bm{s}, r'}$, a sample from $\mu_{\bm{s}, r'}^k$ is obtained by computing $j \bmod p^k$ into an auxiliary register and then measuring the register. More precisely, denote by $\ket{\phi_{\bm{s}, r'}^k(\bm{x})}$ the result of 
    \begin{align}
        \ket{\phi_{\bm{s}, r'}(\bm{x})}\ket{0}
        & \mapsto \frac{1}{\sqrt{r'}} \sum_{j = 0}^{r' - 1} \ket{j}\ket{\bm{x} + j\bm{s}}\ket{j \bmod p^k} \label{equ:r-1}  \\
        & \mapsto \frac{1}{\sqrt{r_k}} \sum_{j = 0}^{r_k - 1} \ket{jp^k + c}\ket{\bm{x} + (jp^k + c)\bm{s}}, \nonumber \tag{measure the last register}
    \end{align}
    where $0 < r_k \le \lfloor r' / p^k \rfloor$ and the random constant $0 \le c \le p^k - 1$ depend on the outcome of measuring the last register. Then
    \[ \rho_{\bm{s}, r'}^k = \E_{\bm{x} \in \U(\Z_q^n)} \Big[ \ket{\phi_{\bm{s}, r'}^k(\bm{x})} \bra{\phi_{\bm{s}, r'}^k(\bm{x})} \Big] \]
    is a sample from $\mu_{\bm{s}, r'}^k$.
    
    \item\label{step:s-mod-p} Computing $\bm{s} \bmod p$: for $k = 0$ we have $j = 0 \bmod p^0$ for all $j$, so $\rho_{\bm{s}, r'}^0 = \rho_{\bm{s}, r'}$ hence $\mu_{\bm{s}, r'}^0 = \mu_{\bm{s}, r'}$. Let $h$ be the smallest integer such that $r' \le p^h$, such an $h$ exists since $r' \le p^e$ by assumption. Then for $k = h$, measuring the last register in \eqref{equ:r-1} collapses the state $\rho_{\bm{s}, r'}$ to $\mathds{1}_{\X'} / (r'q^n)$ where $\X' = \C^{\Z_{r'} \times \Z_q^n}$. Therefore, by a hybrid argument (Lemma \ref{lem:hybrid}) there is a minimal $0 < t \le h$ such that $D$ can distinguish between $\mu_{\bm{s}, r'}^{t - 1}$ and $\mu_{\bm{s}, r'}^t$ with non-negligible advantage. Using the amplification technique of Section \ref{sec:err-red} we can assume that the distinguishing advantage of $D$ is exponentially close to $1$. Note that $t$ can be efficiently computed by analyzing the output of $D$. Let $\bm{s} = (s_1, \dots, s_n)$. We recover $s_1 \bmod p$, the other $s_i \bmod p$ can be recovered similarly. Consider the state $\ket{\phi_{\bm{s}, r'}^{t - 1}(\bm{x})}$ where $\bm{x} = (x_1, \dots, x_n)$, and let $a \in \Z_p$ and $\bm{y} = (y_1, \dots, y_n) \in \Z_q^n$ be arbitrary. If we perform the transform
    \begin{equation}
        \label{equ:s1-trans}
        U_1: \ket{j}\ket{\bm{y}}\ket{0} \mapsto \ket{j}\ket{\bm{y}}\ket{y_1 - ja \bmod p^t}
    \end{equation}
    on $\ket{\phi_{\bm{s}, r'}^{t - 1}(\bm{x})}\ket{0}$ we obtain the state
    \[ \frac{1}{\sqrt{r_{t - 1}}} \sum_{j = 0}^{r_{t - 1} - 1} \ket{jp^{t - 1} + c}\ket{\bm{x} + (jp^{t - 1} + c)\bm{s}}\ket{x_1 + (jp^{t - 1} + c)(s_1 - a) \bmod p^t}. \]
    Measuring the last register results in a state that will be a sample from $\mu_{\bm{s}, r'}^{t - 1}$ or $\mu_{\bm{s}, r'}^t$ depending on whether $s_1 = a$ or $s_1 \ne a \bmod p$, respectively:
    \begin{itemize}
    \item $s_1 = a \bmod p$. In this case, the value of the last register is $x_1 + (s_1 - a)c \bmod p^t$. So the last register is not entangled with the first two registers, and we obtain the original sample from $\mu_{\bm{s}, r'}^{t - 1}$.
        
    \item $s_1 \ne a \bmod p$. Let $0 \le c_1 \le p^k - 1$ be the outcome of the measurement. Then the post-measurement state contains the terms with $j$ satisfying $jp^{t -1} = (c_1 - x_1) / (s_1 - a) - c \bmod p^t$. If we write the right hand side as $c_2p^{t - 1}$ for some constant $0 \le c_2 \le p - 1$ then $j = c_2 \bmod p$ and the post-measurement state is 
    \begin{align*}
        \ket{\psi_{\bm{x}}}
        & = \frac{1}{\sqrt{r_t}} \sum_{j = 0}^{r_t - 1} \ket{(jp + c_2)p^{t - 1} + c}\ket{\bm{x} + ((jp + c_2)p^{t - 1} + c)\bm{s}} \\
        & = \frac{1}{\sqrt{r_t}} \sum_{j = 0}^{r_t - 1} \ket{jp^t + c_2p^{t - 1} + c}\ket{\bm{x} + (jp^t + c_2p^{t - 1} + c)\bm{s}}
    \end{align*}
    where $0 < r_t \le \lfloor r_{t - 1} / p \rfloor$. We clearly have
    \[ \E_{\bm{x} \in \U(\Z_q^n)} \Big[ \ket{\psi_{\bm{x}}}\bra{\psi_{\bm{x}}}  \Big] \in \mu_{\bm{s}, r'}^t. \]
    \end{itemize}
    Therefore, using $D$ we can find out whether $s_1 = a \bmod p$. Since $p \le \poly(n)$, we can recover $s_1 \bmod p$ by trying every $a \in \Z_p$.
    \item\label{step:s-mod-pk} Computing $\bm{s} \bmod p^k$: assume we have recovered, for some $k > 1$, the first $k - 1$ digits of $s_1$ in base $p$, that is we have computed $0 \le \tilde{s}_1 < p^{k - 1}$ such that $\tilde{s}_1 = s_1 \bmod p^{k - 1}$. Let $s_{1, k - 1}$ be the $k$-th digit of $s_1$ in base $p$. To compute $s_{1, k - 1}$, we can modify the transform in Step \ref{step:s-mod-p} as
    \begin{equation*}
        U_k: \ket{j}\ket{\bm{y}}\ket{0} \mapsto \ket{j}\ket{\bm{y}}\ket{y_1 - j\tilde{s}_1 - jp^{k - 1}a \bmod p^{t + k - 1}}.
    \end{equation*}
    Applying $U_k$ to a sample from $\mu_{\bm{s}, r'}^{t - 1}$ and measuring the last register produces a sample from $\mu_{\bm{s}, r'}^{t - 1}$ or $\mu_{\bm{s}, r'}^t$ depending on whether $s_{1, k - 1} = a$ or $s_{1, k - 1} \ne a \bmod p$, respectively. This method works as long as $t + k - 1 \le e$.
    
    \item Quantum Fourier transform: using the procedure in Step \ref{step:s-mod-pk}, we can compute $\bm{s} \bmod p^{e - t + 1}$ where $t$ is  the minimal integer determined in Step \ref{step:s-mod-p}. Similarly, we recover $\bm{s} \bmod p^{e_i - t_i + 1}$ with $t_i$ the corresponding minimal integer for $p_i$, for all $i$. Note that if $r' \le p_i$ then $t_i = 1$. Also, we always have $r' \ge p_i^{t_i - 1}$. Using the Chinese remainder theorem we can compute $\tilde{\bm{s}} = \bm{s} \bmod v$ where
    \[ v = \prod_{i = 1}^\ell p_i^{e_i - t_i + 1} = q / \prod_{i = 1}^\ell p_i^{t_i - 1} \ge \frac{q}{(r')^k} \ge \frac{q}{r}. \]
    Now, by applying the transform $\ket{j}\ket{\bm{y}} \mapsto \ket{j}\ket{\bm{y} - j\tilde{\bm{s}}}$ to the states $\ket{\phi_{\bm{s}, r}(\bm{x})}$, we assume all the coordinates of $\bm{s}$ are multiples of $v$. Next, using Lemma \ref{lem:self-rd}, we project $\ket{\phi_{\bm{s}, r}(\bm{x})}$ onto $\ket{\phi_{\bm{s}, q'}(\bm{x})}$, where $q' = q / v \le r$. Finally, we apply $\qft_{q^n}$ to the second register of $\ket{\phi_{\bm{s}, q'}(\bm{x})}$ and measure to obtain the state
    \[ \frac{1}{\sqrt{q'}} \sum_{j = 0}^{q' - 1} \omega_q^{j\lrang{\bm{u}, \bm{s}}} \ket{j} = \frac{1}{\sqrt{q'}} \sum_{j = 0}^{q' - 1} \omega_{q'}^{j\lrang{\bm{u}, \bm{s} / v}} \ket{j} = \qft_{q'}\ket{\lrang{\bm{u}, \bm{s} / v} \bmod q'}, \]
    where $\bm{u} \in \Z_q^n$ is uniformly random. Applying $\qft_{q'}^*$ to the above state, we obtain $\lrang{\bm{u}, \bm{s} / v} \bmod q'$. The value $\bm{s} / v$ can be computed, with high probability, by gathering $O(n)$ of these linear equations. \qedhere
    \end{enumerate}
\end{proof}

\begin{corollary}
    Let $q = p_1^{e_1} \cdots p_\ell^{e_\ell}$ be the prime factorization of $q$ and assume that the primes $p_i$ are of size $\poly(n)$. If $r \le p_i$ for all $1 \le i \le \ell$ then there is a polynomial-time quantum reduction from solving search-$\edcp_{n, q, r}$ to solving decision-$\edcp_{n, q, r}$. 
\end{corollary}

It follows from the above corollary that for a prime power $q = p^e$, where $p < \poly(n)$ and $r \le p$, search-$\edcp_{n, q, r}$ and decision-$\edcp_{n, q, r}$ are quantum polynomial-time equivalent. Two interesting special cases of this equivalence are $\{q = 2^e, r = 2\}$ and $\{q = p, r < p\}$.


\section{A New Decision Problem}
\label{sec:new-decsn}

In this section, we propose an $\edcp$ decision problem that we believe is more suitable for applications than the decision problem defined in Section \ref{sec:preli}. In particular, our public-key cryptosystem (Section \ref{sec:public-key-enc}) is based on the new decision problem. We show that the new decision problem is quantum polynomial-time equivalent to search-$\edcp$. This establishes the fact that the new decision problem is at least is as hard as the old one.

We assume, as before, that the modulus $q$ has $\poly(n)$-bounded prime factors. Perhaps the new decision problem is best understood for a prime modulus $q$. So let us assume, for now, that $q$ is a $\poly(n)$-bounded prime. Define the distribution $\tilde{\mu}_{\bm{s}, r}$ on the unit sphere $\SX$ by choosing $\bm{x} \in \Z_q^n$ and $t \in \Z_q {\setminus} \{ 0 \}$ uniformly at random and outputting the state
\begin{equation}
    \label{equ:new-dec}
    \ket{\phi_{\bm{s}, r}(\bm{x}, t)} = \frac{1}{\sqrt{r}} \sum_{j = 0}^{r - 1} \omega_q^{jt} \ket{j}\ket{\bm{x} + j\bm{s}}.
\end{equation}
The new decision problem is to distinguish between the distributions $\mu_{\bm{s}, r}$ and $\tilde{\mu}_{\bm{s}, r}$. The motivation behind this new definition is that from states of the form \eqref{equ:new-dec} we can efficiently obtain ``\textit{shifted}'' $\lwe$ samples $(\bm{a}, \lrang{\bm{a}, \bm{s}} + e + t)$ where $\bm{a} \in \Z_q^n$ is uniformly random and $e$ is sampled from $\mathcal{D}_{\Z, q / \lambda}$ for an appropriate value of $\lambda$. For a large enough $q$, this pair is closer to a uniformly random element of $\Z_q^n \times \Z_q$ than an $\lwe$ sample. An instance of the above decision problem then translates to an instance of the $\lwe$ decision problem.

A shifted $\lwe$ sample can be obtained from the state \eqref{equ:new-dec} using the technique in \cite{brakerski2018learning}, which we briefly explain in the following. We ignore the normalization factors in front superpositions for clarity. Applying the transform $\ket{j} \mapsto \ket{j - \lfloor (r - 1) / 2 \rfloor}$ to the first register of the state \eqref{equ:new-dec} we obtain the state
\begin{equation}
    \label{equ:symm-edcp}
    \sum_{j = -\lfloor (r - 1) / 2 \rfloor}^{\lceil (r - 1) / 2 \rceil} \omega_q^{jt} \ket{j}\ket{\bm{x} + j\bm{s}},
\end{equation}
where we have again denoted the uniformly random element $\bm{x} + \lfloor (r - 1) / 2 \rfloor \bm{s} \in \Z_q^n$ by $\bm{x}$. Next, using Lemma \ref{lem:qrs} we can transform \eqref{equ:symm-edcp} into 
\begin{equation}
    \label{equ:symm-eg}
    \sum_{j = -\lfloor (r - 1) / 2 \rfloor}^{\lceil (r - 1) / 2 \rceil} \omega_q^{jt} g_\lambda(j) \ket{j}\ket{\bm{x} + j\bm{s}},
\end{equation}
with probability $\Omega(\lambda / r)$. Here, $g_\lambda(x) = \exp(-\pi x^2 / \lambda^2)$ is a one-dimensional Gaussian function. We assume that $r$ is large enough so that the above probability is not too small. For example, $r = \lfloor \lambda\sqrt{n} \rfloor$ and so $\Omega(\lambda / r) = \Omega(1 / \sqrt{n})$. Now if we apply the transform $\qft_q \otimes \qft_{q^n}$ to \eqref{equ:symm-eg} and measure the last register we obtain the state
\[ \ket{\psi} = \sum_{y \in \Z_q} \sum_{j = -\lfloor (r - 1) / 2 \rfloor}^{\lceil (r - 1) / 2 \rceil} \omega_q^{j(\lrang{\bm{a}, \bm{s}} + y + t)} g_\lambda(j) \ket{y}, \]
where $\bm{a} \in \Z_q^n$ is uniformly random and known. In the following, we use the notation $\ket{\psi_1} \approx \ket{\psi_2}$ when the two quantum states $\ket{\psi_1}$ and $\ket{\psi_2}$ have an exponentially small trace distance. We have
\begin{align*}
    \ket{\psi}
    & \approx \sum_{y \in \Z_q} \sum_{j \in \Z} \omega_q^{j(\lrang{\bm{a}, \bm{s}} + y + t)} g_\lambda(j) \ket{y} \tag{by Corollary \ref{cor:gaus-apprx}} \\
    & = \sum_{y \in \Z_q} \sum_{j \in \Z} g_{1/\lambda} \Big( j + \frac{\lrang{\bm{a}, \bm{s}} + y + t}{q} \Big) \ket{y} \tag{by Theorem \ref{thm:poisson-sum}} \\
    & = \sum_{e \in \Z} g_{1/\sigma} \Big( \frac{e}{q} \Big) \ket{\lrang{-\bm{a}, \bm{s}} + e - t \bmod q} \tag{$e \leftarrow jq + \lrang{\bm{a}, \bm{s}} + t + y$} \\
    & \approx \sum_{e \in \Z_q} g_{1/\sigma} \Big( \frac{e}{q} \Big) \ket{\lrang{-\bm{a}, \bm{s}} + e - t \bmod q} \tag{by Corollary \ref{cor:gaus-apprx}}.
\end{align*}
Measuring the above state, we obtain a pair $(-\bm{a}, \lrang{-\bm{a}, \bm{s}} + e - t)$ where $e$ is sampled from $\mathcal{D}_{\Z, q / \lambda}$.

For a general modulus $q$, an immediate generalization of the above decision problem would be to just replace the prime modulus with a general one, and the above process of obtaining an $\lwe$ sample goes through without any change. However, for such a generalization, it is not clear how to reduce the search problem to the decision problem when $q$ is super-polynomially large in $n$. An alternative approach is to associate a distribution to each of the primes $p \mid q$, then the decision problem is to distinguish between these distribution and $\mu_{\bm{s}, r}$. We make this precise in the following.

\begin{definition}[EDCP, Decision]
    Let $p \mid q$ be a prime. Define the distribution $\mu_{\bm{s}, r, p}$ on the unit sphere $\SX$ by choosing $\bm{x} \in \Z_q^n$ and $t \in \Z_p {\setminus} \{ 0 \}$ uniformly at random and outputting the state
    \begin{equation}
        \label{equ:new-dec1}
        \ket{\phi_{\bm{s}, r}(\bm{x}, p, t)} = \frac{1}{\sqrt{r}} \sum_{j = 0}^{r - 1} \omega_p^{jt} \ket{j}\ket{\bm{x} + j\bm{s}}.
    \end{equation}
    The decision-$\edcp_{n, q, r}$ is the problem of distinguishing between the distribution $\mu_{\bm{s}, r}$ and any distribution in the set $\{ \mu_{\bm{s}, r, p} \}_{p \mid q}$.
\end{definition}

\begin{theorem}
    Assume that all the prime factors of $q$ are $\poly(n)$-bounded. Then there is a quantum polynomial-time reduction from solving search-$\edcp_{n, q, r}$ to solving decision-$\edcp_{n, q, r}$.
\end{theorem}
\begin{proof}
    Let $q = p_1^{e_1} \cdots p_\ell^{e_\ell}$ be the prime factorization of $q$. Let $D$ be an oracle for solving decision-$\edcp_{n, q, r}$. The idea is to use $D$ to find $\bm{s} \bmod p_i^{e_i}$ for all $i$ and then reconstruct $\bm{s} \bmod q$ using the Chinese remainder theorem. Let $\bm{s} = (s_1, \dots, s_n)$. We show how to recover $s_1 \bmod p_1^{e_1}$, the other values $s_i \bmod p_j^{e_j}$ can be recovered similarly. Set $p = p_1$ and $e = e_1$.
    For any $y \in \Z_p$ and nonzero $c \in \Z_p$ define the unitary  
    \[ U_{c, y} \ket{j}\ket{\bm{a}} = \omega_p^{(a_1 - jy)c}\ket{j}\ket{\bm{a}}, \]
    where $a_1$ is the first coordinate of $\bm{a}$. Given a sample $\rho_{\bm{s}, r}$ from $\mu_{\bm{s}, r}$, fix $y \in \Z_p$ and select a fresh nonzero $c \in \Z_p$ uniformly at random. Then we have
    \[ U_{c, y}\ket{\phi_{\bm{s}, r}(\bm{x})} = \frac{1}{\sqrt{r}} \omega_p^{x_1c}\sum_{j = 0}^{r - 1} \omega_p^{j(s_1 - y)c} \ket{j}\ket{\bm{x} + j\bm{s}}. \]
    Therefore, ignoring the global phase, if $s_1 \ne y \bmod p$ then $U_{c, y} \rho_{\bm{s}, r} U_{c, y}^*$ is a sample from $\mu_{\bm{s}, r, p}$, otherwise $U_{c, y} \rho_{\bm{s}, r} U_{c, y}^* = \rho_{\bm{s}, r}$. So, the oracle $D$ could tell us which is the case. Trying all $y \in \Z_p$ we can find $s_1 \bmod p$. Now assume we have recovered $\tilde{s}_1 = s_1 \bmod p^k$ for $k < e$. To compute $s_1 \bmod p^{k + 1}$, we can modify the unitary $U_{c, y}$ as 
    \[ U_{c, y, k} \ket{j}\ket{\bm{a}} = \omega_{p^{k + 1}}^{(a_1 - j\tilde{s}_1 - jp^ky)c}\ket{j}\ket{\bm{a}}. \]
    To see how $U_{c, y, k}$ acts on a sample $\rho_{\bm{s}, r}$ from $\mu_{\bm{s}, r}$, let $s_{1, k}$ be the $(k + 1)$-th digit of $s_1$ in base $p$. Then
    \begin{align*}
        U_{c, y, k} \ket{j}\ket{\bm{x} + j\bm{s}}
        & = \omega_{p^{k + 1}}^{(x_1 + js_1 - j\tilde{s}_1 - jp^ky)c}\ket{j}\ket{\bm{x} + j\bm{s}} \\
        & = \omega_{p^{k + 1}}^{x_1c} \omega_p^{j(s_{1, k} - y)c}\ket{j}\ket{\bm{x} + j\bm{s}}.
    \end{align*}
    Therefore, repeating the above procedure, we can recover $s_{1, k}$. This completes the proof.
\end{proof}


\section{Information-Theoretic and Hardness Bounds}
\label{sec:hardness}

In this section, we derive hardness bounds for $\edcp$ based on the number of samples. The $\edcp_{n, q, r}^m$ problem is to recover the secret $\bm{s} \in \Z_q^n$ given $m$ copies of the state $\rho_{\bm{s}, r}$. We consider bounds $O( n\log q / \log r)$, $\poly(n)$ and $2^{O(\sqrt{n \log q})}$ for $m$. As $m$ increases, $\edcp$ becomes easier. Figure \ref{fig:hardness-bounds} summarizes the hardness of $\edcp$ based on these bounds.  

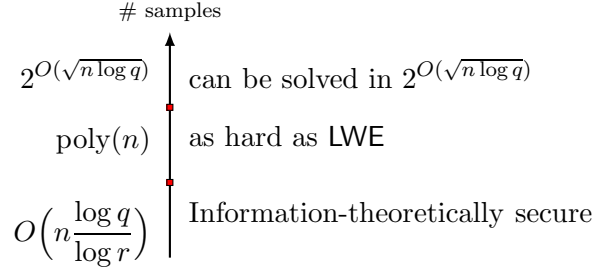
\begin{figure}
    \centering
    \begin{tikzpicture}[tick/.style = {draw = black, fill = red, inner sep = 0.4mm}]
        \draw [thick, -{latex}] (0, 0) -- (0, 3) node[above] {\scriptsize \# samples};
        \node [tick] (t1) at (0, 1) {};
        \node [tick] (t2) at (0, 2) {};

        \node [below left = 1mm of t1] {$\displaystyle O\Big( n \frac{\log q}{\log r} \Big)$};
        \node [below right = 1mm of t1] {Information-theoretically secure};
        \node [below left = 1mm of t2] {$\poly(n)$};
        \node [below right = 1mm of t2] {as hard as $\lwe$};
        \node [above left = 1mm of t2] {$\displaystyle 2^{O(\sqrt{n \log q})}$};
        \node [above right = 1mm of t2] {can be solved in $\displaystyle 2^{O(\sqrt{n \log q})}$};
    \end{tikzpicture}
    \caption{The hardness of $\edcp$ based on the number of samples. These bounds also hold for the security of the public-key encryption scheme of Section \ref{sec:public-key-enc}. In that case, a sample is a copy of the public key.}
    \label{fig:hardness-bounds}
\end{figure}

\subsection{Limited number of samples}
\label{sec:hardness-limited}

We derive a lower bound on $m$ using tools from quantum information theory. A special case of the result of this section was also obtained in \cite{bacon2005optimal}. To better understand the problem we can model it in the following standard way, for $m = 1$. Alice selects a uniformly random $\bm{s} \in \Z_q^n$ and stores it in the classical register $\mathsf{Y}$. She then prepares a register $\mathsf{X}$ in the state $\rho_{\bm{s}, r}$ and sends $\mathsf{X}$ to Bob. So, Bob has access to the register $\mathsf{X}$ that is prepared according to the ensemble
\[
\begin{array}{rrll}
    \eta: & \Z_q^n & \longrightarrow & \mathrm{Pos}(\X) \\
    & \bm{s} & \longmapsto & \frac{1}{q^n}\rho_{\bm{s}, r},
\end{array}
\]
where $\PosSemi(\X)$ is the space of positive semidefinite operators on $\X$. Bob picks a measurement $\mu: \Z_q^n \rightarrow \mathrm{Pos}(\X)$ and measures $\mathsf{X}$ according to $\mu$. He stores the outcome of the measurement in a classical register $\mathsf{Z}$. The pair $(\mathsf{Y}, \mathsf{Z})$ of classical registers will be distributed according to the distribution
\[
\begin{array}{rrll}
    q: & \Z_q^n \times \Z_q^n & \longrightarrow & [0, 1] \\
    & (\bm{s}, \bm{u}) & \longmapsto & \lrang{\mu(\bm{u}), \eta(\bm{s})}.
\end{array}
\]
The amount of information that Bob can learn about $\mathsf{Y}$ using the measurement $\mu$ is determined by the mutual information between $\mathsf{Y}$ and $\mathsf{Z}$, which we denote by $\operatorname{I}_\mu(\eta)$. The accessible information of the ensemble $\eta$ is defined as
\[ \operatorname{I}(\eta) = \sup_{\mu} \operatorname{I}_\mu(\eta), \]
where the supremum ranges over all choices of the measurement $\mu$. Note that the pair $(\mathsf{Y}, \mathsf{X})$ is in the classical-quantum state
\[ \sigma = \sum_{\bm{s} \in \Z_q^n} \ket{\bm{s}}\bra{\bm{s}} \otimes \eta(\bm{s}). \]
The quantum mutual information between $\mathsf{Y}$ and $\mathsf{X}$, with respect to the state $\sigma$, is called the Holevo information of the ensemble $\eta$ and is denoted by $\chi(\eta)$. We have
\begin{equation}
    \label{equ:holevo-chi}
    \chi(\eta) = \operatorname{I}(\mathsf{Y} : \mathsf{X}) = \entpy\Bigg( \frac{1}{q^n} \sum_{\bm{s} \in \Z_q^n} \rho_{\bm{s}, r} \Bigg) - \frac{1}{q^n} \sum_{\bm{s} \in \Z_q^n} \entpy(\rho_{\bm{s}, r}),
\end{equation}
where $\entpy(\cdot)$ is the von Neumann entropy.
\begin{theorem}[Holevo's theorem]
    \label{thm:holevo}
    Let $\Sigma$ be an alphabet and let $\X$ be a complex Euclidean space. For any ensemble $\eta: \Sigma \rightarrow \mathrm{Pos}(\X)$ it holds that $\operatorname{I}(\eta) \le \chi(\eta)$.
\end{theorem}
Therefore, according to Theorem \ref{thm:holevo}, we can obtain an upper bound on the accessible information of the ensemble $\eta$ by computing $\chi(\eta)$. For an arbitrary number of samples $m$, we need to bound $\chi(\eta^{\otimes m})$ for the ensemble
\begin{equation}
    \label{equ:ensem-m}
    \begin{array}{rrll}
        \eta^{\otimes m}: & \Z_q^n & \longrightarrow & \mathrm{Pos}(\X^{\otimes m}) \\
        & \bm{s} & \longmapsto & \frac{1}{q^n}\rho_{\bm{s}, r}^{\otimes m}.
    \end{array}
\end{equation}
This is done is the following theorem.
\begin{theorem}
    \label{thm:acc-bound}
    For the ensemble $\eta^{\otimes m}$ in \eqref{equ:ensem-m} we have $\chi(\eta^{\otimes m}) \le m(1 - q^{-n})\log r$.
\end{theorem}
\begin{proof}
    First, note that the entropy is additive with respect to tensor products, i.e., for any two states $\sigma_1$ and $\sigma_2$ it holds that $\entpy(\sigma_1 \otimes \sigma_2) = \entpy(\sigma_1) + \entpy(\sigma_2)$. It follows that $\entpy(\rho_{\bm{s}, r}^{\otimes m}) = m\entpy(\rho_{\bm{s}, r})$. Next, for a state $\sigma$ of a compound register $(\mathsf{X}_1, \cdots, \mathsf{X}_m)$, we have, by the subadditivity of von Neumann entropy,
    \[ \entpy(\sigma) \le \entpy(\tr_1(\sigma)) + \cdots + \entpy(\tr_m(\sigma))\]
    where $\tr_k(\sigma)$ is the reduction of $\sigma$ to the (state of the) register $\mathsf{X}_k$. Therefore,
    \begin{align*}
        \entpy\bigg( \frac{1}{q^n} \sum_{\bm{s} \in \Z_q^n} \rho_{\bm{s}, r}^{\otimes m} \bigg)
        & \le \sum_{i = 1}^m \entpy\bigg(\tr_i\bigg( \frac{1}{q^n} \sum_{\bm{s} \in \Z_q^n} \rho_{\bm{s}, r}^{\otimes m} \bigg)\bigg) \tag{by subadditivity of $\entpy$} \\
        & = \sum_{i = 1}^m \entpy\bigg(\frac{1}{q^n} \sum_{\bm{s} \in \Z_q^n} \tr_i(\rho_{\bm{s}, r}^{\otimes m}) \bigg) \tag{by linearity of $\tr$} \\
        & = m\entpy\bigg( \frac{1}{q^n} \sum_{\bm{s} \in \Z_q^n} \rho_{\bm{s}, r} \bigg).
    \end{align*}
    It follows from \eqref{equ:holevo-chi}  that $\chi(\eta^{\otimes m}) \le m\chi(\eta)$. Now, we can compute $\chi(\eta)$ by computing the eigenvalues of the operators $\rho_{\bm{s}, r}$ and $\rho = q^{-n}\sum_{\bm{s} \in \Z_q^n} \rho_{\bm{s}, r}$. The eigenvectors of $\rho_{\bm{s}, r}$ are
    \[ \ket{\psi_{\bm{x}, t}} = \frac{1}{\sqrt{r}} \sum_{j = 0}^{r - 1} \omega_r^{jt} \ket{j}\ket{\bm{x} + j\bm{s}}, \hspace*{2mm} (t, \bm{x}) \in \Z_r \times \Z_q^n, \]
    and the eigenvalues are $0$ and $q^{-n}$ with multiplicities $(r - 1)q^n$ and $q^n$, respectively. To compute the eigenvalues of $\rho$ it is best to write the second register in the Fourier basis. We have
    \begin{align*}
        (\mathds{1} \otimes \qft_{q^n}) \rho_{\bm{s}, r} (\mathds{1} \otimes \qft_{q^n})^*
        & = \E_{\bm{x} \in \U(\Z_q^n)} \Big[ (\mathds{1} \otimes \qft_{q^n}) \ket{\phi_{\bm{s}, r}(\bm{x})} \bra{\phi_{\bm{s}, r}(\bm{x})} (\mathds{1} \otimes \qft_{q^n})^* \Big] \\
        & = \frac{1}{rq^n} \sum_{\bm{y}, \bm{z} \in \Z_q^n} \E_{\bm{x} \in \U(\Z_q^n)} \Big[ \omega_q^{\lrang{\bm{x}, \bm{y} - \bm{z}}} \Big] \sum_{j, k \in \Z_r} \omega_q^{\lrang{j\bm{y} - k\bm{z}, \bm{s}}} \ket{j}\bra{k} \otimes \ket{\bm{y}}\bra{\bm{z}} \\
        & = \frac{1}{r} \E_{\bm{y} \in \U(\Z_q^n)} \Big[ \sum_{j, k \in \Z_r} \omega_q^{\lrang{(j - k)\bm{y}, \bm{s}}} \ket{j}\bra{k} \otimes \ket{\bm{y}}\bra{\bm{y}} \Big],
    \end{align*}
    where the last equality follows from the fact that
    \[
    \E_{\bm{x} \in \U(\Z_q^n)} \Big[ \omega_q^{\lrang{\bm{x}, \bm{y} - \bm{z}}} \Big] =
    \begin{cases}
        1 & \text{if } \bm{y} = \bm{z} \\
        0 & \text{if } \bm{y} \ne \bm{z}.
    \end{cases}
    \]
    Therefore, we have 
    \begin{align}
        (\mathds{1} \otimes \qft_{q^n}) \rho (\mathds{1} \otimes \qft_{q^n})^*
        & = \frac{1}{r} \E_{\bm{y} \in \U(\Z_q^n)} \Big[ \sum_{j, k \in \Z_r} \E_{\bm{s} \in \U(\Z_q^n)} \Big[ \omega_q^{\lrang{(j - k)\bm{y}, \bm{s}}} \Big] \ket{j}\bra{k} \otimes \ket{\bm{y}}\bra{\bm{y}} \Big] \nonumber \\
        & = \frac{1}{rq^n} \sum_{j, k \in \Z_r} \ket{j}\bra{k} \otimes \ket{0}\bra{0} + \frac{1}{rq^n} \mathds{1} \otimes (\mathds{1} - \ket{0}\bra{0}) \label{equ:f-basis}
    \end{align}
    where the second equality follows from
    \[
    \E_{\bm{s} \in \U(\Z_q^n)} \Big[ \omega_q^{\lrang{(j - k)\bm{y}, \bm{s}}} \Big] = 
    \begin{cases}
        1 & \text{if } (j - k)\bm{y} = 0 \\
        0 & \text{if } (j - k)\bm{y} \ne 0.
    \end{cases}
    \]
    The eigenvectors of \eqref{equ:f-basis} are
    \begin{align*}
        & \ket{\psi_{\bm{x}, t}} = \frac{1}{\sqrt{r}} \sum_{j = 0}^{r - 1} \omega_r^{jt} \ket{j}\ket{\bm{x}}, \hspace*{2mm} (t, \bm{x}) \in \Z_r \times \Z_q^n, \hspace*{1mm} \bm{x} \ne 0, \\
        & \ket{\psi_t} = \frac{1}{\sqrt{r}} \sum_{j = 0}^{r - 1} \omega_r^{jt} \ket{j}\ket{0}, \hspace*{2mm} t \in \Z_r,
    \end{align*}
    and the eigenvalues are $0$, $q^{-n}$ and $r^{-1}q^{-n}$ with multiplicities $r - 1$, $1$ and $(q^n - 1)r$, respectively. Finally, using \eqref{equ:holevo-chi} and the eigenvalues for $\rho$ and $\rho_{\bm{s}, r}$, we have
    \[ m\chi(\eta) \le m\Big( \frac{1}{q^n}\log(q^n) + \frac{r(q^n - 1)}{rq^n}\log(rq^n) - \log(q^n) \Big) = m\Big( 1 - \frac{1}{q^n} \Big)\log r. \qedhere \]
\end{proof}
Assume that Bob has found a measurement $\mu$ on $\eta^{\otimes m}$, i.e., a measurement that can operate on the joint state of $m$ copies of Alice's state, such that, after possibly some post-measurement processing, he can guess the value of $\mathsf{Y}$ with a constant probability $p$. Then a lower bound on $m$, that depends on $p$, can be computed using Theorem \ref{thm:acc-bound}. We need the following result known as Fano's inequality.
\begin{lemma}[Fano’s inequality]
    \label{lem:fano-ineq}
    Let $X$ and $Y$ be random variables on some finite set $\Gamma$, and let $\tilde{X} = f(Y)$ for some function $f$. Let $p = \Pr[X \ne \tilde{X}]$. Then it holds that
    \[ \entpy(X \vert Y) \le p\log(\abs{\Gamma} - 1) + \entpy(p, 1 - p). \]
\end{lemma}
\begin{corollary}
    \label{cor:lower-b}
    Let $\bm{s} \in \Z_q^n$ be chosen uniformly at random. The number of copies of the state $\rho_{\bm{s}, r}$ needed to recover $\bm{s}$ with constant probability is at least $O(n\log q / \log r)$.
\end{corollary}
\begin{proof}
    Recall the communication scenario above: Alice selects $\bm{s} \in \Z_q^n$ uniformly at random and stores it in the register $\mathsf{Y}$. She then generates $m$ copies of the state $\rho_{\bm{s}, r}$ and sends them to Bob. On receiving $\rho_{\bm{s}, r}^{\otimes m}$, Bob applies a measurement $\mu$ and stores the measurement outcome in the register $\mathsf{Z}$. Bob might perform some post-processing on $\mathsf{Z}$ to obtain another register $\tilde{\mathsf{Z}}$. Assume that $p = \Pr[\mathsf{Y} \ne \tilde{\mathsf{Z}}]$ is a constant. Then
    \begin{align*}
        \operatorname{I}_\mu(\eta^{\otimes m})
        & = \operatorname{I}(\mathsf{X} : \mathsf{Z}) \\
        & = \entpy(\mathsf{Y}) - \entpy(\mathsf{Y} \vert \mathsf{Z}) \tag{by definition} \\
        & \ge n\log(q) - (1 - p)\log(q^n - 1) - \entpy(p, 1 - p) \tag{by Lemma \ref{lem:fano-ineq}} \\
        & \ge pn\log(q) - 1
    \end{align*}
    Now, by Theorem \ref{thm:acc-bound}, $m(1 - q^{-n})\log r \ge \chi(\eta^{\otimes m}) \ge \operatorname{I}_\mu(\eta^{\otimes m})$ which completes the proof.
\end{proof}
\begin{remark}
    \label{rmk:edcp-extreme}
    An interesting case for which the bound in Corollary \ref{cor:lower-b} is tight is when $q = r$. In this case, given the state $\rho_{\bm{s}, r}$, applying the transform $\qft_r^* \otimes \qft_{q^n}$ results in the state
    \[ \frac{1}{\sqrt{q^n}} \sum_{\bm{y} \in \Z_q^n} \omega_q^{\lrang{\bm{y}, \bm{x}}} \ket{\lrang{\bm{y}, \bm{s}}}\ket{\bm{y}}. \]
    Measuring this state, we obtain a linear equation $\lrang{\bm{y}, \bm{s}}$ where $\bm{y} \in \Z_q^n$ is uniformly random. We can solve for $\bm{s}$ by gathering $O(n)$ of these linear equations. 
\end{remark}

\subsection{Polynomial number of samples}
\label{sec:hardness-poly}

When the number of samples is $\poly(n)$, $\edcp$ is quantum polynomially equivalent to $\lwe$ \cite{brakerski2018learning}. The reduction form $\edcp$ to $\lwe$ is proved as in Section \ref{sec:new-decsn}. The reduction from $\lwe$ to $\edcp$ is based on the ball-intersection technique that was originally proposed by \cite{regev2004quantum}. We briefly review the reduction idea here and refer the reader to \cite{regev2004quantum, brakerski2018learning} for details.

Let $(\bm{A}, \bm{b}_0 = \bm{As}_0 + \bm{e}_0)$ be a set of $m$ samples from $\lwe_{n, q, \alpha}$, written in matrix form. We start by preparing the state 
\[ \sum_{\bm{s} \in \Z_q^n} \sum_{j = 0}^{r - 1} \ket{j}\ket{\bm{s}}, \]
where we have omitted the normalization factors for clarity. Here, $r$ is a function of $n$ and $q$. We then compute $(j, \bm{s}) \mapsto \bm{As} - j\bm{b}_0$ into an auxiliary register. After a change of variables we obtain the state
\begin{equation}
    \label{equ:latt-sup}
    \sum_{\bm{s} \in \Z_q^n} \sum_{j = 0}^{r - 1} \ket{j}\ket{\bm{s} + j\bm{s}_0}\ket{\bm{As} - j\bm{e}_0}.
\end{equation}
The goal is to project the above state onto a state $\sum_{\bm{s} \in \Z_q^n} \sum_{j = 0}^{r - 1} \ket{j}\ket{\bm{s} + j\bm{s}_0}$ for some $\bm{s} \in \Z_q^n$ with high probability. To do this, the idea is to draw $m$-dimensional balls around the points $\bm{As} - j\bm{e}_0$ for all $\bm{s} \in \Z_q^n$ and $j \in \Z_r$ and then select a random point in one of these balls. Let $\mathrm{B}_m(0, R)$ be a ball of radius $R$ around $0$. To implement the above idea, we can represent $\mathrm{B}_m(0, R)$ using points of a fine grid. More precisely, $\mathrm{B}_m(0, R)$ is represented by $\tilde{\mathrm{B}}_m(0, R) = \frac{1}{L} \Z^m \cap \mathrm{B}_m(0, R)$ for a large integer $L$. We can efficiently prepare (an approximation of) the superposition
\begin{equation}
    \label{equ:m-ball-sup}
    \ket{\tilde{\mathrm{B}}_m(0, R)} = \frac{1}{\sqrt{\tilde{\mathrm{B}}_m(0, R)}} \sum_{\bm{x} \in \tilde{\mathrm{B}}_m(0, R)} \ket{\bm{x}}.
\end{equation}
Note that for any $\bm{y} \in \Z_q^m$ we have $\bm{y} + \tilde{\mathrm{B}}_m(0, R) = \tilde{\mathrm{B}}_m(\bm{y}, R)$. Tensoring the states in \eqref{equ:latt-sup} and \eqref{equ:m-ball-sup} and adding the third register to the fourth register we obtain the state
\[ \sum_{\bm{s} \in \Z_q^n} \sum_{j = 0}^{r - 1} \ket{j}\ket{\bm{s} + j\bm{s}_0}\ket{\bm{As} - j\bm{e}_0}\ket{\tilde{\mathrm{B}}_m(\bm{As} - j\bm{e}_0, R)}. \]
For an appropriate choice of the radius $R$, for each $\bm{s} \in \Z_q^n$ the intersection $\cap_{j} \tilde{\mathrm{B}}_m(\bm{As} - j\bm{e}_0, R)$ is large, while $\tilde{\mathrm{B}}_m(\bm{As} - j\bm{e}_0, R) \cap \tilde{\mathrm{B}}_m(\bm{As}' - j'\bm{e}_0, R) = \emptyset$ for any $\bm{s} \ne \bm{s}'$ and any $j, j'$. Therefore, if we measure the last register we obtain the state
\[ \sum_{j = 0}^{r - 1} \ket{j}\ket{\bm{s} + j\bm{s}_0}\ket{\bm{As} - j\bm{e}_0} \]
for some random $\bm{s} \in \Z_q^n$, with probability $O(1 - 1 / \ell)$ where $\ell = \poly(n\log q)$. The last register can be uncomputed using the transform $\ket{j}\ket{\bm{x}}\ket{\bm{y}} \mapsto \ket{j}\ket{\bm{x}}\ket{\bm{y} - \bm{Ax} + j\bm{b}_0}$ to obtain the state $\sum_{j = 0}^{r - 1} \ket{j}\ket{\bm{s} + j\bm{s}_0}$.

The above procedure produces an $\edcp$ sample from $\lwe$ samples with a probability that is only polynomially close to $1$. This means we can obtain a polynomial number of $\edcp$ sample from a polynomial number of $\lwe$ samples, and that is the most we can do. In other words, producing a super-polynomial number of $\edcp$ samples from a super-polynomial number of $\lwe$ sample using the above procedure, can be done only with negligible probability. There is no known reduction from $\lwe$ to $\edcp$ for which the sample conversion probability is, for example, subexponentially close to $1$.

\subsection{Subexponential number of samples}
\label{sec:hardness-subexp}

When the number of samples is subexponential, $\edcp$ can be solved in time subexponential in $O(n\log q)$. This can be done using Kuperberg's algorithm \cite{kuperberg2005subexponential, kuperberg2011another}, which solves the hidden subgroup problem for the dihedral group $D_N$. The idea of the algorithm is to use a sieve on states of the form
\begin{equation}
    \label{equ:dih-coset}
    \frac{1}{\sqrt{2}}(\ket{0}\ket{x} + \ket{1}\ket{x + s}),
\end{equation}
where $x \in \Z_N$ is uniformly random, to recover the hidden shift $s \in \Z_N$. The state \eqref{equ:dih-coset} is called a dihedral coset state and the problem of recovering $s$, given such states, is called the Dihedral Coset Problem (DCP). The complexity of Kuperberg's algorithm is $2^{O(\sqrt{\log N})}$ for both time and space. Regev \cite{regev2004subexponential} improved the algorithm to use only $\poly(\log N)$ space at the cost of slightly increasing the running time to $2^{O(\sqrt{\log N \log\log N})}$.

Note that DCP is a special case of $\edcp_{n, q, r}$ where $n = 1$, $q = N$ and $r = 2$. Conversely, $\edcp$ can be reduced to vectorial variant of DCP which can be solved using a similar algorithm as in \cite{kuperberg2005subexponential}. We briefly explain the steps of the algorithm here.
\begin{theorem}
    \label{thm:subexp-smpl}
    Given $2^{O(\sqrt{n\log q})}$ samples, $\edcp_{n, q, r}$ can be solved in time $2^{O(\sqrt{n\log q})}$.
\end{theorem}
\begin{proof}
    Let $\bm{s} = (s_1, \dots, s_n)$. We will recover $s_n$, the rest of the $s_i$ can be recovered similarly. The proof proceeds in a sequence of simple reductions.
    \begin{enumerate}[leftmargin = *, font = \bfseries]
    \item From $\edcp_{n, q, r}$ to DCP over $\Z_q^n$: Given the distribution $\mu_{\bm{s}, r}$ we can efficiently sample from the distribution $\mu_{\bm{s}, 2}$ using Lemma \ref{lem:self-rd}. A sample from $\mu_{\bm{s}, 2}$ is of the form
    \[ \ket{\phi_{\bm{x}, 2}} = \frac{1}{\sqrt{2}}(\ket{0}\ket{\bm{x}} + \ket{1}\ket{\bm{x} + \bm{s}}), \]
    where $\bm{x} \in \Z_q^n$ is uniformly random. This is a dihedral coset state over the group $\Z_q^n$.
    \item From DCP over $\Z_q^n$ to DCP over $\Z_q$: measuring the second register of $(\mathds{1} \otimes \qft_{q^n}) \ket{\phi_{\bm{x}, 2}}$ we obtain the state
    \begin{equation}
        \label{equ:dih-coset-v}
        \ket{\phi_{\bm{y}}} = \frac{1}{\sqrt{2}}(\ket{0} + \omega_q^{\lrang{\bm{y}, \bm{s}}}\ket{1})
    \end{equation}
    where $\bm{y} \in \Z_q^n$ is the outcome of the measurement and is uniformly random. Given two such states $\ket{\phi_{\bm{y}_1}}$ and $\ket{\phi_{\bm{y}_2}}$, we can compute the state
    \[ \ket{\phi_{\bm{y}_1 - \bm{y}_2}} = \frac{1}{\sqrt{2}}(\ket{0} + \omega_q^{\lrang{\bm{y}_1 - \bm{y}_2, \bm{s}}}\ket{1}) \]
    with probability $1 / 2$ by measuring the second register of $\textsc{cnot} \ket{\phi_{\bm{y}_1}} \ket{\phi_{\bm{y}_2}}$. If $\bm{y}_1$ and $\bm{y}_2$ had the first $k$ coordinates in common then $\bm{y}_1 - \bm{y}_2$ would have $0$ in the first $k$ coordinates. From this, we can perform a sieve operation: prepare many states of the form \eqref{equ:dih-coset-v}, then pair the states that have $\bm{y}$ with the same first $k$ coordinates, and then perform the above operation to produce new states with $\bm{y}$ that have the first $k$ coordinates zeroed out. Repeating the same process on the new states produces states with $\bm{y}$ that have first $2k$ coordinates equal to $0$, and so on. The final output of this process is a state $\ket{\phi_{\bm{y}}}$ where $\bm{y} = (0, \dots, 0, y)$, i.e., the state
    \begin{equation}
        \label{equ:dih-sve}
        \ket{\phi_y} := \ket{\phi_{\bm{y}}} = \frac{1}{\sqrt{2}}(\ket{0} + \omega_q^{ys_n}\ket{1}).
    \end{equation}
    This is a DCP state over the group $\Z_q$.
    \item Kuperberg for DCP over $\Z_q$: From the states \eqref{equ:dih-sve} $s_n$ can be recovered using Kuperberg's algorithm.
    \end{enumerate}
    To analyze the above algorithm, suppose we start with $q^\ell$ states. Since we are zeroing out $k$ coordinates at each stage, there are $n / k$ stages. At any stage, if there are $c \cdot q^k$ states, it can be shown, using a simple application of Lemma \ref{lem:hoeffding}, that at least $c / 8 \cdot q^k$ states survive the sieve operation. Therefore, to have $\Theta(q^k)$ states remaining in the last stage, we must have $q^\ell 8^{-n / k} \ge q^k$, hence $\ell \ge k + 3n / (k\log q)$. To minimize the right hand side, we take $k \in \Theta(\sqrt{n / \log q})$, and therefore, we can take $\ell \in \Theta(\sqrt{n / \log q})$.
\end{proof}
\begin{remark}
    When $q = \poly(n)$, Kuperberg's algorithm is not very efficient for solving DCP over $\Z_q$. Instead, we can use a POVM called the Pretty Good Measurement (PGM) \cite{hausladen1994pretty}. Suppose we have prepared the states $\ket{\phi_{y_0}}, \dots, \ket{\phi_{y_t}}$, of the form \eqref{equ:dih-sve}, for some $t \ge \lceil \log q \rceil + 1$. The tensor product of these states is
    \[ \ket{\psi} := \bigotimes_{j = 0}^{t - 1} \frac{1}{\sqrt{2}}(\ket{0} + \omega_q^{y_js_n}\ket{1}) = \frac{1}{\sqrt{2^t}} \sum_{x \in \{0, 1\}^t} \omega_q^{\alpha(x)s_n} \ket{x}, \]
    where $\alpha(x) = x_0y_0 + \cdots + x_ty_t \bmod q$. Using PGM on the state $\ket{\psi}$, we can recover $s_n$ with constant probability \cite{bacon2005optimal}. The implementation of PGM, in this case, boils down to inverting the function $\alpha: \{0, 1\}^t \rightarrow \Z_q$, which can be done efficiently since $q = \poly(n)$.   
\end{remark}


\section{Quantum Public-Key Cryptosystem}
\label{sec:public-key-enc}

A quantum public-key cryptosystem, similar to a classical system, consists of three algorithms:
\begin{itemize}[itemsep = 1pt]
\item $\gen(1^n)$ generates a public-key $pk$ and a secret-key $sk$ based on the security parameter $n$.
\item $\enc(pk, m)$ outputs a ciphertext $c$ for a given public-key $pk$ and message $m$.
\item $\dec(sk, c)$ outputs a message $m$ for a given secret-key $sk$ and ciphertext $c$.
\end{itemize}
The output pair $(pk, sk)$ of the $\gen$ algorithm for a quantum system consists  of a quantum state and a classical state, respectively. In particular, the public-key $pk$ is a quantum state that is generated using a classical key $sk$. The algorithm $\enc$ encrypts the message $m$, which is classical information, using the quantum state $pk$. The output $c$ of $\enc$ is a quantum state. The algorithm $\dec$ uses the key $sk$ to decrypt the quantum state into a classical message $m$.

For the security parameter $n$, we set the parameters for public key system as follows. We choose a prime $p = \poly(n)$ and set $q = p^s$ for some integer $s > 0$. We also set $r = p^{s'}$ where $s' < s$. The reason for these choices of parameters is that the resulting encryption scheme is simpler and more efficient. More generally, one could select $q$ to be a positive integer with $\poly(n)$-bounded prime factors and $r = p^{s'} \vert q$ to be a proper divisor where $p$ is prime. In what follows, we describe our cryptosystem for encrypting a one-bit message $b \in \{ 0, 1 \}$. Since there is only one prime $p$, we drop the parameter $p$ in \eqref{equ:new-dec1} for clarity.

\begin{description}[leftmargin = *]
\item [$\gen(1^n)$:] Select $\bm{x}, \bm{s} \in \Z_q^n$ uniformly at random. Apply the transform $\qft_r \otimes \mathds{1}$ to the register $\ket{0} \ket{\bm{x}}$ to obtain the state $\ket{\psi} = \frac{1}{\sqrt{r}} \sum_{j = 0}^{r - 1} \ket{j} \ket{\bm{x}}$. Apply the transform $A_{\bm{s}}: \ket{j}\ket{\bm{x}} \mapsto \ket{j}\ket{\bm{x} + j\bm{s}}$ to $\ket{\psi}$ to obtain the state
\[ \ket{\phi_{\bm{s}, r}(\bm{x}, 0)} = \frac{1}{\sqrt{r}} \sum_{j = 0}^{r - 1} \ket{j} \ket{\bm{x} + j\bm{s}}. \]
Return the public-key, secret-key pair $(pk, sk) = (\ket{\phi_{\bm{s}, r}(\bm{x}, 0)}, \bm{s})$. 

\item [$\enc(pk = \rho_{\bm{s}, r, 0}, b \in \{ 0, 1 \})$:]  Select $t \in \Z_r {\setminus} \{0\}$ uniformly at random. Apply the transform $U: \ket{j}\ket{\bm{y}} \mapsto \omega_p^{btj}\ket{j}\ket{\bm{y}}$ to $\rho_{\bm{s}, r, 0}$ to obtain the state
\begin{align*}
    U \rho_{\bm{s}, r, 0} U^*
    & = U \E_{\bm{x} \leftarrow \Z_q^n} \Big[ \ket{\phi_{\bm{s}, r}(\bm{x}, 0)} \bra{\phi_{\bm{s}, r}(\bm{x}, 0)} \Big] U^* \\
    & = \E_{\bm{x} \leftarrow \Z_q^n} \Big[ U \ket{\phi_{\bm{s}, r}(\bm{x}, 0)} \bra{\phi_{\bm{s}, r}(\bm{x}, 0)} U^* \Big] \\
    & = \E_{\bm{x} \leftarrow \Z_q^n} \Big[ \ket{\phi_{\bm{s}, r}(\bm{x}, bt)} \bra{\phi_{\bm{s}, r}(\bm{x}, bt)} \Big] \\
    & = \rho_{\bm{s}, r, bt}.
\end{align*}
Return $\rho_{\bm{s}, r, bt}$.

\item [$\dec(sk = \bm{s}, c = \rho_{\bm{s}, r, bt})$:]  Apply the transform $S_{\bm{s}}: \ket{j}\ket{\bm{y}} \mapsto \ket{j}\ket{\bm{y} - j\bm{s}}$ to $\rho_{\bm{s}, r, bt}$. Discard the second register. Apply $\qft_r$ to the resulting state and measure. If the measurement result is 0 then output 0, otherwise output 1. 

\end{description}
\begin{lemma}[Correctness]
    For any bit $b \in \{ 0, 1 \}$ and all outputs $(\bm{s}, \rho_{\bm{s}, r, 0})$ of $\gen$, we have
    \[ \Pr [ \dec(\bm{s}, \enc(\rho_{\bm{s}, r, 0}, b)) = b ] = 1. \]
\end{lemma}
\begin{proof}
    Given the ciphertext $\rho_{\bm{s}, r, tb}$, the decryption steps are as follows
    \begin{align*}
        \rho_{\bm{s}, r, b}
        & \mapsto \E_{\bm{x} \leftarrow \Z_q^n} \Big[ S_{\bm{s}} \ket{\phi_{\bm{s}, r}(\bm{x}, bt)} \bra{\phi_{\bm{s}, r}(\bm{x}, bt)} S_{\bm{s}}^* \Big]  \tag{apply $S_{\bm{s}}$} \\
        & = \E_{\bm{x} \leftarrow \Z_q^n} \bigg[ \frac{1}{r} \sum_{k, j = 0}^{r - 1} \omega_p^{bt(j - k)}\ket{k}\ket{\bm{x}} \bra{j}\bra{\bm{x}} \bigg] \\
        & = \frac{1}{r} \sum_{k, j = 0}^{r - 1} \omega_p^{bt(j - k)}\ket{k}\bra{j} \otimes \E_{\bm{x} \leftarrow \Z_q^n} \Big[ \ket{\bm{x}}\bra{\bm{x}} \Big] \\
        & \mapsto \frac{1}{r} \sum_{k, j = 0}^{r - 1} \omega_p^{bt(j - k)}\ket{k}\bra{j} \tag{discard the second register} \\
        & \mapsto \qft_r \frac{1}{r} \sum_{k, j = 0}^{r - 1} \omega_p^{bt(j - k)}\ket{k}\bra{j} \qft_r^* \tag{apply quantum Fourier transform}\\
        & = \ket{btr/p} \bra{btr/p}
    \end{align*}
    If $b = 0$ then the above state is $\ket{0}\bra{0}$, otherwise it is $\ket{tr / p}\bra{tr / p} \ne \ket{0}\bra{0}$. 
\end{proof}

\paragraph{Discussion.}
The above encryption scheme can be naturally based on the original decision-$\edcp$ (Definition \ref{def:d-edcp}) as well. The resulting scheme, however, is not as efficient. To see this, suppose public keys are generated using the $\gen$ algorithm above. Bob encrypts a bit $b \in \{0, 1\}$ as follows: if $b = 0$ then he outputs the public key as the ciphertext. If $b = 1$ then he measures the first register of the public key to obtain a state $\ket{j}\ket{\bm{x}}$, for uniformly random $(j, \bm{x}) \in \Z_r \times \Z_q^n$, and outputs this state as the ciphertext. Distinguishing between the encryption of $0$ and $1$ is then equivalent to solving the decision-$\edcp$. Now, if Alice runs the above $\dec$ algorithm on the ciphertext she gets two possible outputs depending on the value of $b$:
\begin{itemize}
\item When $b = 0$, the output is always $0$, since the last state obtained in the algorithm is always $\ket{0}\bra{0}$ and so the measurement outcome is $0$.
\item When $b = 1$, the output is $1$ with probability $1 - 1 / r$. This is because the last state obtained in the algorithm, just before the last measurement, is $\qft_r \ket{j}\bra{j} \qft_r^*$. 
\end{itemize}
To decrease the above (rather large) decryption error down to, say, $2^{-\Omega(n)}$, the $\dec$ algorithm has to be repeated $\Omega(n / \log r)$ times. In the quantum setting, that means Alice has to have access to $\Omega(n / \log r)$ copies of the ciphertext, which in turn means Bob needs the same number of copies of the public key to generate the ciphertexts. This is equivalent to saying that for an encryption-decryption round with negligible error probability, the size of the public key increases by a factor of $\Omega(n / \log r)$. This scheme then has no advantage over a classical $\lwe$-based encryption scheme.

\subsection{Circuits}

All three algorithms $\gen, \enc, \dec$ in the above public-key cryptosystem are very easy to implement. In the following, we briefly describe the circuits for these algorithms. Since all the arithmetic unites used in the circuits are already known, we will not give a gate-level design for them.

Figure \ref{fig:gen-circuit} shows the key generation circuit. The gate $\qft_r$ is the quantum Fourier transform over $\Z_r$, and the gate $A_{\bm{s}}$ is the multiply-add transform $\ket{j}\ket{\bm{x}} \mapsto \ket{j}\ket{\bm{x} + j\bm{s}}$.

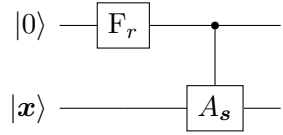
\begin{figure}[h]
    \centering
    \begin{quantikz}[thin lines]
        \lstick{$\ket{0}$} & \gate{\qft_r} & \ctrl{1} & \qw \\
        \lstick{$\ket{\bm{x}}$} & \qw  & \gate{A_{\bm{s}}} & \qw 
    \end{quantikz}
    \caption{The circuit for $\gen$}
    \label{fig:gen-circuit}
\end{figure}

For the encryption algorithm, we need to implement the transform $\ket{j}\ket{\bm{y}} \mapsto \omega_p^{btj}\ket{j}\ket{\bm{y}}$. We start by preparing the state
\[ (\mathds{1} \otimes \qft_p) \frac{1}{\sqrt{r}} \sum_{j = 0}^{r - 1} \ket{j} \ket{\bm{x} + j\bm{s}} \otimes \ket{1}  = \frac{1}{\sqrt{r}} \sum_{j = 0}^{r - 1} \ket{j} \ket{\bm{x} + j\bm{s}} \frac{1}{\sqrt{p}} \sum_{z \in \Z_p} \omega_p^z \ket{z}. \]
Then we apply the transform $T_{bt}: \ket{j}\ket{\bm{y}}\ket{z} \mapsto \ket{j}\ket{\bm{y}}\ket{z - jbt}$ to obtain the state
\[ \frac{1}{\sqrt{r}} \sum_{j = 0}^{r - 1} \ket{j} \ket{\bm{x} + j\bm{s}} \frac{1}{\sqrt{p}} \sum_{z \in \Z_p} \omega_p^z \ket{z - jbt} = \frac{1}{\sqrt{r}} \sum_{j = 0}^{r - 1} \omega_p^{btj} \ket{j} \ket{\bm{x} + j\bm{s}} \frac{1}{\sqrt{p}} \sum_{z \in \Z_p} \omega_p^z \ket{z}. \]
Finally, we measure the last register to obtain the desired state. Figure \ref{fig:enc-circuit} shows the encryption circuit.

\begin{figure}[h]
    \centering
    \begin{quantikz}[thin lines]
         \lstick{$\ket{j}$} & \qw & \ctrl{2} & \qw & \qw \\
         \lstick{$\ket{\bm{y}}$} & \qw & \qw & \qw & \qw \\
         \lstick{$\ket{1}$} & \gate{\qft_p} & \gate{T_{bt}} & \meter{} & \qw
    \end{quantikz}
    \caption{The circuit for $\enc$.}
    \label{fig:enc-circuit}
\end{figure}
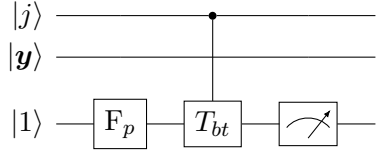

Figure \ref{fig:dec-circuit} shows the decryption circuit. The gate $\qft_r$ is the quantum Fourier transform over $\Z_r$, and the gate $S_{\bm{s}}$ is the multiply-subtract operation $\ket{j}\ket{\bm{x}} \mapsto \ket{j}\ket{\bm{x} - j\bm{s}}$.

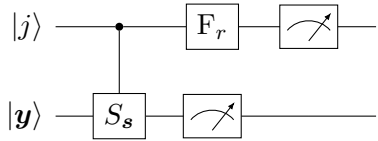
\begin{figure}[h]
    \centering
    \begin{quantikz}[thin lines]
        \lstick{$\ket{j}$} & \ctrl{1} & \gate{\qft_r} & \meter{} & \qw \\
        \lstick{$\ket{\bm{y}}$} & \gate{S_{\bm{s}}} & \meter{} & \qw & \qw
    \end{quantikz}
    \caption{The circuit for $\dec$.}
    \label{fig:dec-circuit}
\end{figure}

\subsection{A very efficient instantiation}

Although our cryptosystem is efficient even for a super-polynomial modulus $q = p^s$ and any $\poly(n)$-bounded prime $p$, it can be made more efficient by choosing a $\poly(n)$-bounded $q$ and a small prime $p$. In particular, we can choose $p = 2$ and $q = 2^s$ such that $q = \poly(n)$. In this case, we choose $r = 2^{s'}$ where $s' \ll s$.

For the above parameters, we have $\omega_p = -1$. In the encryption algorithm, since $t \ne 0$, the only choice for $t$ is $t = 1$ and, therefore, random number generation is not required. For an input bit $b$ the ciphertext state is
\[ \ket{\phi_{\bm{s}, r}(\bm{x}, b)} = \frac{1}{\sqrt{r}} \sum_{j = 0}^{r - 1} (-1)^{bj} \ket{j}\ket{\bm{x} + j\bm{s}}. \]
The phase $(-1)^{bj}$ is much simpler to compute than the more general phase $\omega_p^{btj}$. In particular, the quantum Fourier transform $\qft_p$ is now the Hadamard transform $H: \ket{x} \mapsto (\ket{0} + (-1)^x \ket{1}) / \sqrt{2}$, and the transform $T_{bt}$ is now $T_b: \ket{j}\ket{\bm{y}}\ket{z} \mapsto \ket{j}\ket{\bm{y}}\ket{z \oplus (jb \bmod 2)}$. Figure \ref{fig:enc-circuit-2} shows the new encryption circuit.

\begin{figure}[h]
    \centering
    \begin{quantikz}[thin lines]
         \lstick{$\ket{j}$} & \qw & \ctrl{2} & \qw & \qw \\
         \lstick{$\ket{\bm{y}}$} & \qw & \qw & \qw & \qw \\
         \lstick{$\ket{1}$} & \gate{H} & \gate{T_b} & \meter{} & \qw
    \end{quantikz}
    \caption{The circuit for $\enc$.}
    \label{fig:enc-circuit-2}
\end{figure}
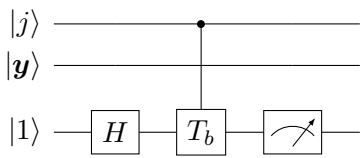

Let us briefly analyze the complexity of this scheme. The $\gen$ algorithm requires random generation, scalar multiplication and addition in $\Z_q^n$. These can be done at the cost of $O(n\log q \log\log q)$ qubit operations. The quantum Fourier transform over $\Z_r$ can be done in $O(\log r\log\log r)$ qubit operations \cite{hales2000improved}. The $\enc$ algorithm takes $O(1)$ since the $T_b$ operation takes $O(1)$. The $\dec$ algorithm has the same complexity as the $\gen$ algorithm.


\bibliographystyle{plain}
\bibliography{references}

\begin{thebibliography}{10}

\bibitem{aaronson2009quantum}
Scott Aaronson.
\newblock Quantum copy-protection and quantum money.
\newblock In {\em 2009 24th Annual IEEE Conference on Computational
  Complexity}, pages 229--242. IEEE, 2009.

\bibitem{aaronson2012quantum}
Scott Aaronson and Paul Christiano.
\newblock Quantum money from hidden subspaces.
\newblock In {\em Proceedings of the forty-fourth annual ACM symposium on
  Theory of computing}, pages 41--60, 2012.

\bibitem{babai2016graph}
L{\'a}szl{\'o} Babai.
\newblock Graph isomorphism in quasipolynomial time.
\newblock In {\em Proceedings of the forty-eighth annual ACM symposium on
  Theory of Computing}, pages 684--697, 2016.

\bibitem{bacon2005optimal}
Dave Bacon, Andrew~M Childs, and Wim van Dam.
\newblock Optimal measurements for the dihedral hidden subgroup problem.
\newblock {\em arXiv preprint quant-ph/0501044}, 2005.

\bibitem{bennett1984quantum}
C.~H. Bennett and G.~Brassard.
\newblock Quantum cryptography: Public key distribution and coin tossing.
\newblock In {\em Proceedings of IEEE International Conference on Computers,
  Systems \& Signal Processing, Bangalore, India}, pages 175--179, 1984.

\bibitem{bennett1983quantum}
Charles~H Bennett and Gilles Brassard.
\newblock Quantum cryptography and its application to provably secure key
  expansion, public-key distribution, and coin-tossing.
\newblock In {\em IEEE International Symposium on Information Theory},
  volume~95. St-Jovite: Qebec Press, 1983.

\bibitem{boneh2015graduate}
Dan Boneh and Victor Shoup.
\newblock {\em A graduate course in applied cryptography}.
\newblock 2020.
\newblock Draft version 0.5.

\bibitem{brakerski2018learning}
Zvika Brakerski, Elena Kirshanova, Damien Stehl{\'e}, and Weiqiang Wen.
\newblock Learning with errors and extrapolated dihedral cosets.
\newblock In {\em IACR International Workshop on Public Key Cryptography},
  pages 702--727. Springer, 2018.

\bibitem{brakerski2013classical}
Zvika Brakerski, Adeline Langlois, Chris Peikert, Oded Regev, and Damien
  Stehl{\'e}.
\newblock Classical hardness of learning with errors.
\newblock In {\em Proceedings of the forty-fifth annual ACM symposium on Theory
  of computing}, pages 575--584, 2013.

\bibitem{childs2010quantum}
Andrew~M Childs and Wim Van~Dam.
\newblock Quantum algorithms for algebraic problems.
\newblock {\em Reviews of Modern Physics}, 82(1):1, 2010.

\bibitem{colbeck2007impossibility}
Roger Colbeck.
\newblock Impossibility of secure two-party classical computation.
\newblock {\em Physical Review A}, 76(6):062308, 2007.

\bibitem{farhi2012quantum}
Edward Farhi, David Gosset, Avinatan Hassidim, Andrew Lutomirski, and Peter
  Shor.
\newblock Quantum money from knots.
\newblock In {\em Proceedings of the 3rd Innovations in Theoretical Computer
  Science Conference}, pages 276--289, 2012.

\bibitem{gottesman2001quantum}
Daniel Gottesman and Isaac Chuang.
\newblock Quantum digital signatures.
\newblock {\em arXiv preprint quant-ph/0105032}, 2001.

\bibitem{hales2000improved}
Lisa Hales and Sean Hallgren.
\newblock An improved quantum fourier transform algorithm and applications.
\newblock In {\em Proceedings 41st Annual Symposium on Foundations of Computer
  Science}, pages 515--525. IEEE, 2000.

\bibitem{hausladen1994pretty}
Paul Hausladen and William~K Wootters.
\newblock A ‘pretty good’measurement for distinguishing quantum states.
\newblock {\em Journal of Modern Optics}, 41(12):2385--2390, 1994.

\bibitem{hoeffding1963probability}
Wassily Hoeffding.
\newblock Probability inequalities for sums of bounded random variables.
\newblock {\em Journal of the American Statistical Association},
  58(301):13--30, 1963.

\bibitem{kawachi2005computational}
Akinori Kawachi, Takeshi Koshiba, Harumichi Nishimura, and Tomoyuki Yamakami.
\newblock Computational indistinguishability between quantum states and its
  cryptographic application.
\newblock In {\em Annual International Conference on the Theory and
  Applications of Cryptographic Techniques}, pages 268--284. Springer, 2005.

\bibitem{kawachi2012computational}
Akinori Kawachi, Takeshi Koshiba, Harumichi Nishimura, and Tomoyuki Yamakami.
\newblock Computational indistinguishability between quantum states and its
  cryptographic application.
\newblock {\em Journal of cryptology}, 25(3):528--555, 2012.

\bibitem{kaye2007introduction}
Phillip Kaye, Raymond Laflamme, Michele Mosca, et~al.
\newblock {\em An introduction to quantum computing}.
\newblock Oxford university press, 2007.

\bibitem{kuperberg2005subexponential}
Greg Kuperberg.
\newblock A subexponential-time quantum algorithm for the dihedral hidden
  subgroup problem.
\newblock {\em SIAM Journal on Computing}, 35(1):170--188, 2005.

\bibitem{kuperberg2011another}
Greg Kuperberg.
\newblock Another subexponential-time quantum algorithm for the dihedral hidden
  subgroup problem.
\newblock {\em arXiv preprint arXiv:1112.3333}, 2011.

\bibitem{lo1997quantum}
Hoi-Kwong Lo and Hoi~Fung Chau.
\newblock Is quantum bit commitment really possible?
\newblock {\em Physical Review Letters}, 78(17):3410, 1997.

\bibitem{lo1999unconditional}
Hoi-Kwong Lo and Hoi~Fung Chau.
\newblock Unconditional security of quantum key distribution over arbitrarily
  long distances.
\newblock {\em science}, 283(5410):2050--2056, 1999.

\bibitem{lydersen2010hacking}
Lars Lydersen, Carlos Wiechers, Christoffer Wittmann, Dominique Elser, Johannes
  Skaar, and Vadim Makarov.
\newblock Hacking commercial quantum cryptography systems by tailored bright
  illumination.
\newblock {\em Nature photonics}, 4(10):686--689, 2010.

\bibitem{mayers1997unconditionally}
Dominic Mayers.
\newblock Unconditionally secure quantum bit commitment is impossible.
\newblock {\em Physical review letters}, 78(17):3414, 1997.

\bibitem{mayers2001unconditional}
Dominic Mayers.
\newblock Unconditional security in quantum cryptography.
\newblock {\em Journal of the ACM (JACM)}, 48(3):351--406, 2001.

\bibitem{micciancio2012trapdoors}
Daniele Micciancio and Chris Peikert.
\newblock Trapdoors for lattices: Simpler, tighter, faster, smaller.
\newblock In {\em Annual International Conference on the Theory and
  Applications of Cryptographic Techniques}, pages 700--718. Springer, 2012.

\bibitem{nikolopoulos2008applications}
Georgios~M Nikolopoulos.
\newblock Applications of single-qubit rotations in quantum public-key
  cryptography.
\newblock {\em Physical Review A}, 77(3):032348, 2008.

\bibitem{nikolopoulos2009deterministic}
Georgios~M Nikolopoulos and Lawrence~M Ioannou.
\newblock Deterministic quantum-public-key encryption: forward search attack
  and randomization.
\newblock {\em Physical Review A}, 79(4):042327, 2009.

\bibitem{ozols2013quantum}
Maris Ozols, Martin Roetteler, and J{\'e}r{\'e}mie Roland.
\newblock Quantum rejection sampling.
\newblock {\em ACM Transactions on Computation Theory (TOCT)}, 5(3):1--33,
  2013.

\bibitem{peikert2017pseudorandomness}
Chris Peikert, Oded Regev, and Noah Stephens-Davidowitz.
\newblock Pseudorandomness of ring-lwe for any ring and modulus.
\newblock In {\em Proceedings of the 49th Annual ACM SIGACT Symposium on Theory
  of Computing}, pages 461--473, 2017.

\bibitem{regev2004quantum}
Oded Regev.
\newblock Quantum computation and lattice problems.
\newblock {\em SIAM Journal on Computing}, 33(3):738--760, 2004.

\bibitem{regev2004subexponential}
Oded Regev.
\newblock A subexponential time algorithm for the dihedral hidden subgroup
  problem with polynomial space.
\newblock {\em arXiv preprint quant-ph/0406151}, 2004.

\bibitem{regev2005lattices}
Oded Regev.
\newblock On lattices, learning with errors, random linear codes, and
  cryptography.
\newblock In {\em Proceedings of the thirty-seventh annual ACM symposium on
  Theory of computing}, pages 84--93, 2005.

\bibitem{regev2009lattices}
Oded Regev.
\newblock On lattices, learning with errors, random linear codes, and
  cryptography.
\newblock {\em Journal of the ACM (JACM)}, 56(6):1--40, 2009.

\bibitem{shor2000simple}
Peter~W Shor and John Preskill.
\newblock Simple proof of security of the bb84 quantum key distribution
  protocol.
\newblock {\em Physical review letters}, 85(2):441, 2000.

\bibitem{watrous2009zero}
John Watrous.
\newblock Zero-knowledge against quantum attacks.
\newblock {\em SIAM Journal on Computing}, 39(1):25--58, 2009.

\bibitem{watrous2018theory}
John Watrous.
\newblock {\em The theory of quantum information}.
\newblock Cambridge University Press, 2018.

\bibitem{wiesner1983conjugate}
Stephen Wiesner.
\newblock Conjugate coding.
\newblock {\em ACM Sigact News}, 15(1):78--88, 1983.

\end{thebibliography}


\newpage
\appendix

\section{Poisson Summation}
Let $G$ be a locally compact abelian group, and let $\mathbb{T}$ be the circle group. The dual group $\hom(G, \mathbb{T})$ of all continuous group homomorphisms from $G$ to $\mathbb{T}$ is denoted by $\widehat{G}$. The operation in $\widehat{G}$ is pointwise multiplication, i.e., for $\chi_1, \chi_2 \in \widehat{G}$, $(\chi_1 . \chi_2)(x) = \chi_1(x)\chi_2(x)$. There is a natural topology on $\widehat{G}$, called the compact-open topology, that makes it a topological group. It can be shown that $\widehat{G}$ is a locally compact abelian group as well. In the representation theory language, $\widehat{G}$ is called the character group of $G$ and the element of $\widehat{G}$ are called characters. 

The group $G$ carries a Haar measure that unique up to a multiplicative positive constant. The space $L^1(G)$ is then defined according to the Haar measure. For a character $\chi \in \widehat{G}$, the Fourier transform of a function $f \in L^1(G)$ is defined by the Haar integral
\[ \hat{f}(\chi) = \int_{G} f(g)\overline{\chi(g)} dg. \]
Let $H \le G$ be a closed subgroup, so there is an exact sequence
\begin{equation}
    \label{equ:ses}
    0 \rightarrow H \rightarrow G \rightarrow G / H \rightarrow 0
\end{equation}
of topological groups. Applying the functor $\hom(-, \mathbb{T})$ to the above sequence, we obtain the exact sequence
\begin{equation}
    \label{equ:d-ses}
    0 \rightarrow \widehat{G / H} \rightarrow \widehat{G} \rightarrow \widehat{H} \rightarrow 0
\end{equation}
of duals. The Fourier transform is a linear map from the groups in \eqref{equ:ses} to the groups in \eqref{equ:d-ses}. The Poisson summation formula relates these Fourier transforms. We can always choose Haar measures on $H, G$, and $G / H$ such that the quotient integral identity
\[ \int_G f(g) dg = \int_{G / H} \int_H f(gh) dh d(gH) \]
holds for every compactly supported continuous function $f: G \rightarrow \C$.
\begin{theorem}[Poisson summation]
    \label{thm:poisson-sum}
    Let $H \le G$ be a closed subgroup and let $f \in L^1(G)$. Define $f_H \in L^1(G / H)$ by $f_H(gH) = \int_H f(gh) dh$. Then $\widehat{f_H} = \hat{f} \vert_{\widehat{G / H}}$, where $\vert$ is restriction. If $\hat{f} \vert_{\widehat{G / H}} \in L^1(\widehat{G / H})$ then we also have
    \[ \int_{H} f(gh) dh = \int_{\widehat{G / H}} \hat{f}(\chi) \chi(g) d\chi \]
    for almost all $g \in G$.
\end{theorem}
An example of a locally compact group is $G = \R^n$ with Haar measure taken to be the usual Lebesgue measure. For each $\bm{u} \in \R^n$ the mapping $\chi_{\bm{u}}: \R^n \rightarrow \C$ defined by $\chi_{\bm{u}}(\bm{x}) = e^{2\pi i \lrang{\bm{u}, \bm{v}}}$ is a character of $G$. In fact, all the elements of $\widehat{G}$ are of this form, and we have $\R^n \simeq \widehat{\R^n}$ via the map $\bm{u} \mapsto \chi_{\bm{u}}$. Let $H = L$ where $L$ is a lattice. Then since $L$ and $\widehat{\R^n / L} = L^\perp$ are both discrete groups, the integrals in Theorem \ref{thm:poisson-sum} are just summations, so we obtain
\[ \sum_{\bm{x} \in L} f(\bm{x} + \bm{y}) = \frac{1}{\mathrm{Vol}(\R^n / L)}\sum_{\bm{x} \in L^\perp} \hat{f}(\bm{x})e^{2\pi i \lrang{\bm{x}, \bm{y}}}  \]


\section{Some Tools From Quantum Information}

\subsection{Rejection sampling}

Let $\bm{\pi} \in \R^n$ be such that $\opnorm{\bm{\pi}}_2 = 1$ and let $\bm{\varepsilon} \in \R^n$ be any vector such that $\bm{\varepsilon} \le \bm{\pi}$, i.e., $\varepsilon_k \le \pi_k$ for all $k$. Let
\begin{equation}
    \label{equ:qreject-1}
    \ket{\psi_{\bm{\pi}}} = \sum_{k = 1}^n \pi_k\zeta_k \ket{k}\ket{\alpha_k},
\end{equation}
where $\abs{\zeta_k} = 1$ for all $k$, and $\alpha_k$ is a function of $k$. The process of transforming the state \eqref{equ:qreject-1} into the state
\begin{equation}
    \label{equ:qreject-2}
    \ket{\psi_{\bm{\varepsilon}}} = \frac{1}{\opnorm{\bm{\varepsilon}}_2} \sum_{k = 1}^n \varepsilon_k\zeta_k \ket{k}\ket{\alpha_k},
\end{equation}
is called quantum rejection sampling \cite{ozols2013quantum}. Define the set of single-qubit operations
\[ R_{\bm{\varepsilon}}(k) = \frac{1}{\pi_k}
\left[
\begin{array}{ll}
    \sqrt{\abs{\pi_k}^2 - \varepsilon_k^2} & -\varepsilon_k \\
    \varepsilon_k &\sqrt{\abs{\pi_k}^2 - \varepsilon_k^2}
\end{array}
\right]
, \quad 1 \le k \le n,
\]
and let $R_{\bm{\varepsilon}} = \sum_{k = 1}^n \ket{k}\bra{k} \otimes \mathds{1} \otimes R_{\bm{\varepsilon}}(k)$. Then we have
\[ R_{\bm{\varepsilon}} \ket{\psi_{\bm{\pi}}}\ket{0} = \sum_{k = 1}^n \zeta_k \ket{k}\ket{\alpha_k} \Big( \sqrt{\abs{\pi_k}^2 - \varepsilon_k^2}\ket{0} + \varepsilon_k\ket{1} \Big). \]
If we measure the last register, the probability of obtaining the state \eqref{equ:qreject-2} is $\sum_{k = 1}^n \abs{\zeta_k\varepsilon}^2 = \lVert \bm{\varepsilon} \rVert_2^2$. We state this result in the following lemma for the sake of reference.
\begin{lemma}
    \label{lem:qrs}
    The state conversion $\ket{\psi_{\bm{\pi}}} \mapsto \ket{\psi_{\bm{\varepsilon}}}$ can be done with probability $\opnorm{\displaystyle \bm{\varepsilon}}_2^2$.
\end{lemma}

\subsection{Norms and Gaussians}

Let $\X = \C^n$. For a quantum state $\sigma \in \mathrm{D}(\X)$, the trace norm is defined as
\[ \opnorm{\sigma}_1 = \tr\big(\sqrt{\sigma^*\sigma}\big). \]
This is the Schatten $p$-norm for $p = 1$. Let $\bm{u}, \bm{v} \in \SX$ be unit vectors, and let
\[ \ket{\psi_{\bm{u}}} = \sum_{k = 1}^n u_k \ket{\alpha_k}, \quad \ket{\psi_{\bm{v}}} = \sum_{k = 1}^n v_k \ket{\alpha_k} \]
be quantum states, where $\{\ket{\alpha_k}\}$ is an orthonormal set, and $v_k$ and $u_k$ are the coordinates of $\bm{u}$ and $\bm{v}$, respectively. Then the following inequality holds between the trace norm and the $\ell_1$ norm.
\begin{equation}
    \label{equ:l1-trace}
    \opnorm{\ket{\psi_{\bm{u}}}\bra{\psi_{\bm{u}}} - \ket{\psi_{\bm{v}}}\bra{\psi_{\bm{v}}}}_1 \le \opnorm{\bm{u} - \bm{v}}_1.
\end{equation}
The inequality \eqref{equ:l1-trace} can be used to approximate quantum states that involve Gaussians. In particular, let $g_r(x) = \exp(-\pi x^2 / r^2)$ be a one-dimensional Gaussian. For any set of complex numbers $\{\zeta_k\}_{k \in \Z}$ on the unit circle and any subset $A \subseteq \Z$ define
\[ g_r(A) = \sum_{k \in A}g_r(k), \quad  \tilde{g}_r(A) = \sum_{k \in A} \zeta_k g_r(k). \]
\begin{lemma}
    \label{lem:trfr-bound-c}
    For any $\kappa, r > 0$, we have $\abs{\tilde{g}_r(\Z {\setminus} [-\sqrt{\kappa}r, \sqrt{\kappa}r])} \le 2^{-\Omega(\kappa)} g_r(\Z)$. 
\end{lemma}
\begin{proof}
    \begin{align*}
        \abs{\tilde{g}_r(\Z {\setminus} [-\sqrt{\kappa}r, \sqrt{\kappa}r])}
        & \le \sum_{k \in \Z {\setminus} [-\sqrt{\kappa}r, \sqrt{\kappa}r]} \abs{\zeta_k g_r(k)} \\
        & = g_r(\Z {\setminus} [-\sqrt{\kappa}r, \sqrt{\kappa}r]) \\
        & \le 2^{-\Omega(\kappa)} g_r(\Z),
    \end{align*}
    where the last inequality follows from Lemma \ref{lem:trfr-bound}.
\end{proof}
\begin{lemma}[{\cite[Lemma 1]{brakerski2018learning}}]
    \label{lem:trfr-bound}
    For any $\kappa, r > 0$, we have $g_r(\Z {\setminus} [-\sqrt{\kappa}r, \sqrt{\kappa}r]) \le 2^{-\Omega(\kappa)} g_r(\Z)$. 
\end{lemma}
\begin{corollary}
    \label{cor:gaus-apprx}
    Let $\kappa, r > 0$ and let $A = [-\sqrt{\kappa}r, \sqrt{\kappa}r]$. Let $\{ \zeta_{k, j} \}_{(k, j) \in \Z_n \times \Z}$ be complex numbers on the unit circle. For the quantum states  
    \[ \ket{\psi_A} = \frac{1}{S} \sum_{k = 1}^n \sum_{j \in A} \zeta_{k, j} g_r(j) \ket{k}, \quad \ket{\psi_\Z} = \frac{1}{T}\sum_{k = 1}^n \sum_{j \in \Z} \zeta_{k, j} g_r(j) \ket{k} \]
    we have $\opnorm{\ket{\psi_A}\bra{\psi_A} - \ket{\psi_\Z}\bra{\psi_\Z}}_1 \le 2^{-\Omega(\kappa)} T^{-1} g_r(\Z)$.
\end{corollary}
\begin{proof}
    This follows from \eqref{equ:l1-trace}, Lemma \ref{lem:trfr-bound-c} and a straightforward calculation.
\end{proof}

\end{document}